\documentclass[11pt, letterpaper]{article}
\title{New Separations Results for External Information}
\author{Mark Braverman
\thanks{Department of Computer Science, Princeton University.
Research supported in part by the NSF Alan T. Waterman Award, Grant No. 1933331, a Packard Fellowship in Science and Engineering, and the Simons Collaboration on Algorithms and Geometry. }
\and
Dor Minzer
\thanks{Department of Mathematics, Massachusetts Institute of Technology.}
}
\date{\vspace{-5ex}}
\usepackage{fullpage}
\usepackage{amsthm}
\usepackage{amsmath,amssymb,amsfonts,nicefrac}
\usepackage{xspace}
\usepackage{color}
\usepackage{url}
\usepackage{hyperref}
\usepackage{bm}
\usepackage{bbm}
\usepackage{times}

\usepackage{enumitem}

\newtheorem{theorem}{Theorem}[section]
\newtheorem{thm}{Theorem}[section]

\newtheorem{lemma}[thm]{Lemma}
\newtheorem{corollary}[thm]{Corollary}
\newtheorem{claim}[thm]{Claim}

\newtheorem{definition}[thm]{Definition}

\newtheorem{fact}[thm]{Fact}

\newtheorem{openproblem}[thm]{Open Problem}

\newcommand\E{\mathop{\mathbb{E}}}
\newcommand\card[1]{\left| {#1} \right|}
\newcommand\sett[2]{\left\{ \left. #1 \;\right\vert #2 \right\}}

\newcommand\set[1]{{\left\{ #1 \right\}}}
\newcommand\Prob[2]{{\Pr_{#1}\left[ {#2} \right]}}

\newcommand\cProb[3]{{\Pr_{#1}\left[ \left. #3 \;\right\vert #2 \right]}}

\newcommand\Expect[2]{{\mathop{\mathbb{E}}_{#1}\left[ {#2} \right]}}

\newcommand\power[1]{\set{0,1}^{#1}}

\newcommand\half{{1\over2}}

\newcommand\defeq{\stackrel{def}{=}}
\newcommand\skipi{{\vskip 10pt}}

\newcommand\eps{\varepsilon}

\renewcommand\geq{\geqslant}
\renewcommand\ge{\geqslant}
\renewcommand\leq{\leqslant}
\renewcommand\le{\leqslant}

\newcommand{\IC}{\mathrm{IC}}

\newcommand{\ZAC}{\mathrm{ZAC}}

\newcommand\MI{\mathrm{I}}
\newcommand{\CC}{\mathrm{CC}}
\newcommand{\HH}{\mathrm{H}}

\newcommand{\DKL}[2]{\mathrm{D}_{\text{KL}}\left( #1 \parallel #2 \right)}

\newcommand{\Ber}{\mathsf{B}}
\newcommand{\rom}[1]{\uppercase\expandafter{\romannumeral #1\relax}}

\makeatletter
\newtheorem*{rep@theorem}{\rep@title}
\newcommand{\newreptheorem}[2]{%
\newenvironment{rep#1}[1]{%
\def\rep@title{\bf #2 \ref*{##1} \text{(Restated)} }%
\begin{rep@theorem} }%
{\end{rep@theorem} } }

\newtheorem*{rep@claim}{\rep@title}
\newcommand{\newrepclaim}[2]{%
\newenvironment{rep#1}[1]{%
\def\rep@title{\bf #2 \ref*{##1} \text{(Restated)} }%
\begin{rep@claim} }%
{\end{rep@claim} } }

\newtheorem*{rep@lemma}{\rep@title}
\newcommand{\newreplemma}[2]{%
\newenvironment{rep#1}[1]{%
\def\rep@title{\bf #2 \ref*{##1} \text{(Restated)} }%
\begin{rep@lemma} }%
{\end{rep@lemma} } }

\makeatother
\newreptheorem{theorem}{Theorem}
\newrepclaim{claim}{Claim}
\newreplemma{lemma}{Lemma}

\begin{document}
\maketitle
\begin{abstract}
We obtain new separation results for the two-party external information complexity of boolean functions. The external information complexity of a
function $f(x,y)$ is the minimum amount of information a two-party protocol computing $f$ must reveal to an outside observer about the input. We obtain the following results:
\begin{itemize}
\item
We prove an exponential separation between external and internal information complexity, which is the best possible; previously no separation was known.
\item
We prove a near-quadratic separation between amortized zero-error communication
complexity and external information complexity for total functions, disproving a conjecture
of \cite{Bravermansurvey}.
\item
We prove a matching upper showing that our separation result is tight.
\end{itemize}
\end{abstract}
\section{Introduction}

\label{sec:intro}

The main object of study in this paper is the {\em external} two-party information complexity of problems. For a two-party communication protocol $\pi(x,y)$, with inputs distributed according to some $(x,y)\sim\mu$, one can define the following complexity measures:\footnote{See Section~\ref{sec:prelim} for rigorous definitions and further background.}
\begin{itemize}
	\item
	the (average case) communication cost $\CC_{\mu}[\pi]$ of $\pi$ is the expected {\em number} of bits exchanged in $\pi$;
	\item
	the {\em external} information cost $\MI^{\sf external}_{\mu}[\pi]$ of $\pi$ is the expected amount of information {\em an external observer} learns about the inputs $(x,y)$ by observing an execution of $\pi(x,y)$;
	\item
	the {\em internal} information cost $\MI^{\sf internal}_{\mu}[\pi]$ of $\pi$ is the expected amount of information {\em the protocol participants} learn about the inputs $(x,y)$ by observing $\pi$.
\end{itemize}
For a given computational task --- such as computing a boolean function $f(x,y)$ with a prescribed error $\eps$ --- one can define the \{average case communication, external information, internal information\} complexity of performing the task by taking the infimum of the corresponding cost over all  protocols $\pi$ that succeed at performing the task. Thus:
\begin{itemize}
	\item
	the (average case) communication complexity $\CC_{\mu}[f,\eps]$  is the smallest expected {\em number} of bits that need to be exchanged to compute $f$ with error $\le \eps$;
	\item
	the {\em external} information complexity $\IC^{\sf external}_{\mu}[f,\eps]$ is the smallest amount of information that must be revealed to an external observer by two parties who need to compute $f$ with error $\le \eps$;
	\item
	the {\em internal} information complexity $\IC^{\sf internal}_{\mu}[f,\eps]$\footnote{Sometimes called simply ``information complexity".} is the smallest amount of information that must be revealed by the players to each other while computing $f$ with error $\le\eps$.
\end{itemize}
It follows from basic information-theoretic calculations that for all tasks:
\begin{equation}\label{eq:ICCC1}
\IC^{\sf internal}_{\mu}[f,\eps]\le
\IC^{\sf external}_{\mu}[f,\eps]\le
\CC_{\mu}[f,\eps].
\end{equation}
Within the theoretical computer science literature, notions of information complexity had been introduced at least twice. It was introduced once in the context of information-theoretic security \cite{bar1993privacy}, where a low $\IC^{\sf internal}_{\mu}[f,\eps]$ would mean that information-theoretically secure two-party computation is possible (it turn out to be impossible in most cases).
It was introduced in a different set of works, starting with the use of external information cost, in \cite{chakrabarti2001informational,bar2004information}, in the context of proving communication complexity lower bounds by using information-theoretic reasoning to prove an information complexity lower bound, and then using \eqref{eq:ICCC1} to deduce a communication
complexity lower bound. More recent surveys on information complexity can be found in \cite{braverman2015interactive,weinstein2015information}.

Starting with the works of Shannon in the 1940s, the main motivation for using information-theoretic quantities is that they {\em tensorize}\footnote{Also
said to satisfy a ``direct sum" property in the TCS literature.}. For example, if $X_1$ and $X_2$ are two independent random variables, then their Shannon's entropy satisfies $H(X_1X_2)=H(X_1)+H(X_2)$. Turns out that  internal information complexity satisfies a similar property (and, thus, is arguably the correct information-theoretic version of two-party communication complexity). For simplicity, denote by $(f^q,\eps)$ the task of computing $q$ independent copies of $f$, each with error $\le \eps$, then
\begin{equation}
	\label{eq:ICCC2}
	\IC^{\sf internal}_{\mu^q}[f^q,\eps] =
	q\cdot \IC^{\sf internal}_{\mu}[f,\eps].
\end{equation}
Such a relationship was known to be false for $\CC$, and was believed to also not hold for $\IC^{\sf external}$, although to the best of our knowledge only in the present paper we rule it out for all values of $\eps$ including $\eps=0$.

Equations \eqref{eq:ICCC1} and \eqref{eq:ICCC2} together provide a blueprint for proving communication lower bound on computing multiple copies of a function:
\begin{equation}
	\label{eq:ICCC3}
	\CC_{\mu^q}[f^q,\eps] \ge  	\IC^{\sf internal}_{\mu^q}[f^q,\eps] =
	q\cdot \IC^{\sf internal}_{\mu}[f,\eps].
\end{equation}
Thus, an information complexity lower bound on $f$, implies a communication lower bound on multiple copies of $f$. Moreover, a slight twist on \eqref{eq:ICCC3} allows one to use similar reasoning to prove lower bounds
on e.g. an OR of $q$ copies of $f$, to obtain tight bounds on the communication complexity of functions such as ${\sf Disjointness}$ \cite{bar2004information,braverman2013information}.

It turns out that, in fact, for $\eps>0$ \eqref{eq:ICCC3} is tight \cite{BravermanRao}. For all $f$ and $\eps>0$ the following holds:
\begin{equation}
	\label{eq:ICCC4}
\lim_{q\rightarrow \infty}	\CC_{\mu^q}[f^q,\eps]/q  =  \IC^{\sf internal}_{\mu}[f,\eps].
\end{equation}
Equation~\eqref{eq:ICCC4} has given rise to two questions:
\begin{itemize}
	\item {\bf What is the relationship between $ \IC^{\sf internal}_{\mu}[f,\eps]$	 and $\CC_{\mu}[f,\eps]$?} How large can the gap in \eqref{eq:ICCC1} be? This question is sometimes called the ``interactive compression question", and is equivalent to the direct sum question for two-party communication complexity;
	\item {\bf What happens in \eqref{eq:ICCC4} when $\eps=0$ --- that is, when no error is allowed?} Note that, not coincidentally, in other communication settings allowing error drastically alters the communication complexity of problems. For example, the communication cost of ${\sf EQ}_n(x,y)$ --- the problem of determining whether two $n$-bit strings are equal is $O(\log 1/\eps)$ independent of $n$ when error $\eps$ is allowed, but increases to $n+1$ when no error is allowed.
\end{itemize}

\paragraph{Separation between information and communication.} The first question was answered by Ganor, Kol, and Raz \cite{GKR}, who showed an exponential separation between internal information complexity and communication complexity. Moreover, such separation is the best possible \cite{braverman2015interactive}. In other words, there is a boolean function $f$ such that
$ \IC^{\sf internal}_{\mu}[f,\eps]=O(k)$, while
$\CC_{\mu}[f,\eps]=\Omega(2^k)$. In addition, exponential separation was shown between external information and communication complexity --- at least for tasks
\cite{ganor2019exponential}.

Therefore, in the context of \eqref{eq:ICCC1}, at least for tasks, the second inequality was known to be strict (with the maximal possible exponential separation). A separation in the first inequality had been strongly suspected but never proven. In this paper (Theorem~\ref{thm:external_internal}) we show that, in fact, the example from \cite{GKR} has an exponential external information complexity (and not just exponential communication complexity), and thus gives an example of an $f$ such that
\begin{equation}
\label{eq:ICCC5}
\IC^{\sf external}_{\mu}[f,\eps]\geq 2^{\Omega(\IC^{\sf internal}_{\mu}[f,\eps])}.
\end{equation}
As will be discussed later, we need \eqref{eq:ICCC5} in order to separate
external information from zero-error amortized communication.

\paragraph{Zero-error amortized communication.} A second mystery that remains
in the wake of \eqref{eq:ICCC4} is what happens with zero-error amortized communication? In other words, what can we say about the quantity:\footnote{Here, zero-error means that the protocol has to always output the correct value
of $f(x,y)$, even if $\mu(x,y)=0$.}
\begin{equation}
	\label{eq:z1}
	\lim_{q\rightarrow \infty}	\CC_{\mu^q}[f^q,0]/q
\end{equation}
A canonical example of a function where zero-error and vanishing-error communication costs diverge is the $n$-bit {\em Equality} function
${\sf EQ}_n(x,y):={\bf 1}_{x=y}$. It is known that the amortized communication
complexity of ${\sf EQ_n}$ is $O(1)$ \cite{feder1995amortized}. As a consequence of this result (which can also be seen directly \cite{braverman2015interactive}), one gets for all $\mu$,
$$
\IC^{\sf internal}_{\mu}[{\sf EQ}_n,0] = O(1).
$$
On the other hand, it is not hard to see using fooling sets, that
$$
	\lim_{q\rightarrow \infty}	\CC_{\mu^q}[{\sf EQ}_n^q,0]/q=\Omega(n),
$$
where $\mu$ is the distribution $\mu=\frac{1}{2} U_{(x,x)} + \frac{1}{2} U_{(x,y)}$ --- a mixture of the uniform distribution and the uniform distribution on ${\sf EQ}^{-1}_n(1)$. Therefore \eqref{eq:ICCC4} has no chance of holding when $\eps=0$. More precisely, half of the proof of \eqref{eq:ICCC4} holds for $\eps=0$, yielding
\begin{equation}
\lim_{q\rightarrow \infty}	\CC_{\mu^q}[f^q,0]/q  \ge  \IC^{\sf internal}_{\mu}[f,0],
\label{eq:z2}
\end{equation}
but this inequality may be strict, as is indeed the case for $f={\sf EQ}_n$.

An attempt to prove the $\le$ direction in \eqref{eq:z2} would involve trying to compress a low-internal information protocol for $f^q$ into a low-communication
one. Such compression procedures exist \cite{BravermanRao}, but they inherently introduce
errors (where with a tiny probability the message received doesn't match the message sent). In contrast to internal information, there is a zero-error compression protocol for {\em external} information \cite{HJMR} leading to
a variant of a converse to \eqref{eq:z2}  where internal information is
replaced with external information:
\begin{equation}
	\lim_{q\rightarrow \infty}	\CC_{\mu^q}[f^q,0]/q  \le  \IC^{\sf external}_{\mu}[f,0].
	\label{eq:z3}
\end{equation}
We formally prove \eqref{eq:z3} for completeness purposes in
Section~\ref{sec:missing} (Theorem~\ref{thm:upper_bd_easy}).
Inequality \eqref{eq:z3}, along with the fact that inequality \eqref{eq:z2} is
easily seen to not be tight led to the following conjecture \cite{Bravermansurvey}:
\begin{flalign}
& {\bf \text  Conjecture:~}	\lim_{q\rightarrow \infty}	\CC_{\mu^q}[f^q,0]/q  =  \Theta(\IC^{\sf external}_{\mu}[f,0]). &
	\label{eq:z4}
\end{flalign}
There were several reasons to believe this conjecture. Translated to this language, a result of Ahlswede and Cai \cite{ahlswede1994communication} shows that \eqref{eq:z4} holds (with a constant $1$) when $f$ is the $2$-bit AND function for the hardest distribution $\mu$ --- the quantity on both sides is $\log_2 3$. A version of \eqref{eq:z4} in fact holds for one-sided non-deterministic communication/information, which we prove in Section~\ref{sec:missing2} for completeness.
Here the $1$-superscript in $\IC^{\sf external,1}_{\mu}$ stands for the complexity of proving that the value of $f(x,y)=1$.

\begin{thm}
	\label{thm:oneside} Let $\mu$ be a distribution with ${\sf supp}(\mu)\subseteq f^{-1}(1)$, then $$
	\IC^{\sf external,1}_{\mu}[f,0]\le
		\lim_{q\rightarrow \infty}	\CC_{\mu^q}^{1^q}[f^q,0]/q.
	$$
\end{thm}

We should note that a quadratic upper bound on $\IC^{\sf external}_{\mu}[f,\eps]$ in terms of amortized zero-error communication complexity (Theorem~\ref{thm:quadratic_tight}) does hold. Informally, this upper bound can be thought of as a consequence of Theorem~\ref{thm:oneside} (along its co-nondeterministic counterpart) similarly to the $D(f)\le N^0(f)\cdot N^1(f)$ bound on deterministic communication complexity in terms of non-deterministic communication complexity.

It should be noted that in Conjecture~\eqref{eq:z4} it is important that the zero-error of communication holds for all potential inputs to $f$ (even when $\mu$ is not full-support). In other words, correctness shouldn't be predicated on a ``promise" about the inputs $(x,y)$.
 In the promise setting, a counterexample  has been constructed by Kol, Moran, Shpilka, and Yehudayoff  \cite{KMSY}. In the context
 of the counterexample, Theorem~\ref{thm:oneside} also doesn't hold, which suggests a large gap between the promise and non-promise regimes.

Our main contribution is to disprove Conjecture~\eqref{eq:z4}. In light of \eqref{eq:z2}, a prerequisite for disproving the conjecture is being able to separate internal information complexity from
external information complexity along the lines of \eqref{eq:ICCC5}.

%
%
%

\subsection{Main results}
We now state our main results formally. First, as alluded to before, we show an exponential separation between internal information and external information.
\begin{thm}\label{thm:external_internal}
  For all $\eps>0$, for large enough $k$, there is $n\in\mathbb{N}$, a function $f\colon\power{n}\times\power{n}\to\power{}$
  and an input distribution $\mu$ satisfying:
  \begin{enumerate}
    \item $\IC^{\sf internal}_{\mu}[f,\eps]\leq O(k)$,
    \item $\IC^{\sf external}_{\mu}[f,\eps]\geq 2^{\Omega(k)}$.
  \end{enumerate}
\end{thm}

Secondly, using Theorem~\ref{thm:external_internal} we disprove Conjecture~\ref{eq:z4}.
We show that even if one considers the external information of $f$ for protocols with constant
error, a near-quadratic gap between it and the amortized zero-error communication complexity is still possible:
\begin{thm}\label{thm:main}
  For large enough $k$, there is $n\in\mathbb{N}$, a function $f\colon\power{n}\times\power{n}\to\power{}$ and an input distribution $\mu$ such that
  \begin{enumerate}
    \item $\lim_{q\rightarrow \infty} \frac{1}{q} \CC_{\mu^q}[f^q,0]\leq O(\sqrt{k}\log^2 k)$,
    \item $\IC^{\sf external}_{\mu}[f,1/16]\geq \Omega(k)$.
  \end{enumerate}
\end{thm}

If one insists on external information of protocols with zero-error, a much stronger separation result holds
(and in fact quickly follows from Theorem~\ref{thm:main}):
\begin{corollary}\label{corr:zero_error_external_separation}
    For large enough $m$, there is $n\in\mathbb{N}$, a function $f\colon\power{n}\times\power{n}\to\power{}$ and an input distribution $\nu$ such that
  \begin{enumerate}
    \item $\lim_{q\rightarrow \infty} \frac{1}{q} \CC_{\nu^q}[f^q,0]\leq O(1)$,
    \item $\IC^{\sf external}_{\nu}[f,0]\geq \Omega(m)$.
  \end{enumerate}
\end{corollary}

It is worth noting that the separation given in Theorem~\ref{thm:main} is nearly tight, and in general
the external information with constant error is at most the square of the amortized zero-error
communication complexity:
\begin{thm}\label{thm:quadratic_tight}
  There exists an absolute constant $C>0$, such that for any $\eps>0$, $f\colon\power{n}\times\power{n}\to\power{}$ and an input distribution $\mu$,
  we have that
  \[
    \IC^{\sf external}_{\mu}[f,\eps]\leq \frac{C}{\eps^2} \left(\lim_{q\rightarrow \infty} \frac{1}{q} \CC_{\mu^q}[f^q,0]\right)^2.
  \]
\end{thm}

\subsection{Proof overview and discussion}

\paragraph{Separating external and internal information.} As discussed earlier, having a problem with a low internal information complexity but a high external information complexity appears essential for our main separation result. We actually prove that the Bursting Noise Function, which was introduced by~\cite{GKR} to separate internal information complexity from randomized communication complexity also separates internal information complexity from external information complexity. To that end, we extend the reach of the relative-discrepancy lower bound technique from \cite{GKR} to apply to external information complexity.

We do this by introducing a property of protocols having ``universally low external information". We then show that (1) low external information cost protocols can be approximated with universally low external information protocols; and (2)
universally low external  information protocols (just like low-communication protocols in \cite{GKR}) do very poorly when trying to compute functions with the appropriate relative discrepancy property.

\paragraph{Separating amortized zero-error communication from external information.} It is actually surprisingly difficult to construct a candidate function for our main separation. We need a function with a low amortized zero-error communication, but a high external information complexity. Low amortized communication implies low internal information complexity. This means that inside the construction we should use a separation between internal and external information.

In addition, as seen in Theorem~\ref{thm:oneside}, non-deterministic zero-error amortized communication complexity appears to be connected to zero-error non-deterministic external information complexity. This connection is reinforced by Theorem~\ref{thm:quadratic_tight}. Therefore, the construction is likely to require a function featuring maximum possible (i.e. quadratic) separation between non-deterministic communication and deterministic communication complexity.

Indeed, our starting point is a boolean function whose query complexity exhibits
a quadratic separation between deterministic and non-deterministic complexity.
To ``lift'' this separation into the communication world, our construction uses an idea that is close in spirit to the cheat sheet lifting constructions~\cite{Cheatsheet1,Cheatsheet2}. At a high level, we too want to give advantage (say, access to non-deterministic certificates) to certain
class of protocols (in our case, protocols that solve many independent instances). Our construction however implements this high-level idea differently.

We hide extra information as an output of an auxiliary function $f$. Our chosen function
$f$ has low internal information complexity so as to be useful for constructing a protocol with low amortized communication complexity.
This extra information allows us to evaluate our function on all but
$o(1)$ fraction of the input tuples of the players in which they fail (and therefore importantly do not introduce errors), and for those
one may use a trivial protocol (which would contribute at most $o(1)$ to the amortized communication anyway). The second property that we need
is that this extra information would be useless for protocols with low external information attempting to solve only a single challenge. Indeed,
our chosen function $f$ will have high external information complexity, and our argument in the proof of Theorem~\ref{thm:external_internal} actually shows that
low external information protocols are unable to gain even a slight advantage for computing $f$. In particular, the extra information we
hide looks essentially random to them, and therefore does not offer any help.

With this intuition in mind, the starting point of our construction is a function $h:\{0,1\}^m\rightarrow \{0,1\}$ and its AND-lifting
$h_{\land}\colon \{0,1\}^m\times\{0,1\}^m\rightarrow \{0,1\}$ defined as $h_{\land}(x,y) = h(x\land y)$,
with the following properties:
(1) the external information of $h_{\land}$ is $\Omega(m)$;
(2) one can certify that $h(z)=0$ or that $h(z)=1$ using only $C = O(\sqrt{m})$ bits. We then wish to construct a function
$H$ on $4$-tuples, $(x,y,u,v)$, whose output on $(x,y,u,v)$ is $h_{\land}(x,y)$, and $(u,v)$ encodes a certificate for that
on the support of our input distribution. Here, by ``encodes'' we mean that $(u,v)$ could be viewed as a sequence of $C$ input tuples
to $f$, and that these bits encode a certificate to $h_{\land}(x,y)$. Thus, for the purposes of amortized communication complexity,
we use the back-door $(u,v)$ and only have to pay communication proportional to the internal information of $f$ (after retrieving the hint $f(u,v)$ we still need to use $O(C)$ communication to verify its veracity to ensure the final answer is never wrong). On the other hand,
$f(u,v)$ will look almost random to a low external information protocols, and hence is essentially useless for them,
so the external information of our protocol must be the external information of $h_{\land}$, i.e. at least $\Omega(m)$.

\paragraph{Discussion.} Of the two questions raised in the beginning of  Section~\ref{sec:intro}, in this paper we have given the optimal separation between internal and external information complexity. On the other hand, the mystery of understanding amortized zero-error communication complexity has only deepened. We now know that it is different from external information complexity, and that the worst possible gap is in some sense quadratic.

This leaves the question of characterizing zero-error amortized communication complexity wide open.

\begin{openproblem}
	Characterize the amortized zero-error communication complexity of functions.
	For simplicity, suppose $\mu$ has full support. Characterize:
	\begin{equation}
		 \ZAC(f,\mu):= \lim_{q\rightarrow\infty} { \CC}_{\mu^q}(f^q,0)/q,
	\end{equation}
where $\CC$ is average-case distributional communication complexity with zero error.
\end{openproblem}

Ideally, the characterization would be in terms of information-theoretic quantities pertaining to computing a single copy of $f$. Such a characterization is sometimes called ``single letter" characterization in the information theory literature. It is likely that understanding this quantity will lead to further
 communication complexity insights.

\paragraph{Organization.}
In Section~\ref{sec:prelim}, we recall some standard notions and tools that will be needed in our proofs.
We prove Theorem~\ref{thm:external_internal} in Section~\ref{sec:int_ext}, and Theorem~\ref{thm:main} in Section~\ref{sec:ammortized_separation}.
Finally, we prove Theorem~\ref{thm:quadratic_tight} and Corollary~\ref{corr:zero_error_external_separation} in Section~\ref{sec:tight}.

\section{Preliminaries}\label{sec:prelim}
\subsection{Information theory}
We begin with a few basic definitions from information theory. Throughout the paper, we only consider random variables with finite support.
\begin{definition}
  Let $X,Y$ be random variables with a finite support.
  \begin{enumerate}
    \item The Shannon entropy of $X$
  is $\HH[X] = \sum\limits_{x}{\Prob{}{X = x}\log\left(\frac{1}{\Prob{}{X = x}}\right)}$.
    \item The Shannon entropy of $X$ conditioned on $Y$ is
    $\HH[X~|~Y] = \Expect{y\sim Y}{\HH[X~|~Y = y]}$, where
    $\HH[X~|~Y = y] = \sum\limits_{x}{\cProb{}{Y=y}{X = x}\log\left(\frac{1}{\cProb{}{Y=y}{X = x}}\right)}$.
  \end{enumerate}
\end{definition}

\begin{definition}
Let $X,Y,Z$ be random variables with a finite support.
  \begin{enumerate}
    \item The mutual information between $X$ and $Y$ is
    $\MI[X;Y] = \HH[X] - \HH[X|Y]$.
    \item The mutual information between $X,Y$ conditioned on $Z$ is
    $\MI[X;Y~|~Z] = \HH[X~|~Z] - \HH[X~|~Y,Z]$.
  \end{enumerate}
\end{definition}

\begin{definition}
Let $X,Y$ be random variables with a finite support. The KL-divergence from $Y$ to $X$ is
    $\DKL{X}{Y} = \sum\limits_{x,y} \Prob{}{X = x}\log\left(\frac{\Prob{}{X = x}}{\Prob{}{Y = y}}\right)$.
\end{definition}

We will need the following standard facts from information theory (for proofs, see~\cite{elements} for example).
\begin{fact}\label{fact:mutual_div_KL}
  Let $X,Y,Z$ be random variables, and let $P_{X,Y}$ be the joint distribution of $X,Y$. Then
  \[
  \MI[ X,Y ; Z] = \Expect{(x,y)\sim P_{X,Y}}{\DKL{Z|_{X=x,Y=y}}{Z}}.
  \]
\end{fact}

\begin{fact}\label{fact:MI_chainrule}
  Let $X,Y_1,\ldots,Y_n$ be random variables. Then
  \[
  \MI[ X; Y_1,\ldots,Y_n] = \sum\limits_{i=1}^{n} \MI[X; Y_i~|~Y_{<i}].
  \]
\end{fact}

\begin{fact}\label{fact:MI_cond}
  Let $X,Y,Z$ be random variables. Then
  $\MI[ X; Y~|~Z]\leq \MI[X; Y,Z]$.
\end{fact}

For $p\in[0,1]$, we denote by $\Ber(p)$ a Bernoulli random variable with parameter $p$.
\begin{fact}\label{fact:KL_bits}
  Let $p,q\in[0,1]$ and suppose that $\frac{1}{3}\leq q\leq \frac{2}{3}$. Then
  \[
  2(p-q)^2
  \leq
  \DKL{\Ber(p)}{\Ber(q)}
  \leq \frac{9(p-q)^2}{2\ln 2}.
  \]
\end{fact}

\subsection{Communication complexity}
Let $f:\mathcal{X}\times\mathcal{Y}\to\power{}$ be a function, let $\mu$ be a distribution over its inputs and denote $(X,Y)\sim \mu$.
Throughout, we denote by $\Pi$ a two-player communication protocol, and by $\Pi(X,Y)$ the distribution over transcripts of the protocol
where the inputs are sampled according to the random variables $(X,Y)$. We denote the output of a specific transcript $\pi$ by
${\sf output}(\pi)$. Abusing notations, for inputs $x,y$, we denote by ${\sf output}(\Pi(x,y))$ the random variable which is the output
of $\Pi$ when run on inputs $x,y$.
\begin{definition}
  The internal information of the protocol $\Pi$ is defined as
  $\MI^{\sf internal}_{\mu}[\Pi] = \MI[\Pi; X | Y]+\MI[\Pi; Y | X]$.

  For an error parameter $0\leq \eps<\half$, we define the internal information cost of
  $f$ on $\mu$ with error $\eps$ by
  \[
  \IC^{\sf internal}_{\mu}[f,\eps] = \inf_{\Pi: \Prob{(x,y)\sim \mu}{f(x,x)\neq {\sf output}(\Pi(x,y))}\leq \eps} \MI^{\sf internal}_{\mu}[\Pi].
  \]
\end{definition}

\begin{definition}
  The external information of the protocol $\Pi$ is defined as
  $\MI^{\sf external}_{\mu}[\Pi] = \MI[\Pi; X,Y]$.

  For an error parameter $0\leq \eps<\half$, we define the external information cost of
  $f$ on $\mu$ with error $\eps$ by
  \[
  \IC^{\sf external}_{\mu}[f,\eps] = \inf_{\Pi: \Prob{(x,y)\sim \mu}{f(x,y)\neq \Pi(x,y)}\leq \eps} \MI^{\sf external}_{\mu}[\Pi].
  \]
\end{definition}

\begin{fact}\label{fact:external_greater_internal}
   For any function $f$ and $\eps>0$ it holds that
  $\IC^{\sf external}_{\mu}[f,\eps]\geq \IC^{\sf internal}_{\mu}[f,\eps]$.
\end{fact}

We need to use the notion of smooth protocols, as defined in~\cite{BBCR}.
\begin{definition}\label{def:smooth}
  A two-player protocol $\Pi$ is called smooth if for every pair of inputs $x,y$, a step $i$ in the protocol,
  and a possible transcript $T$ up to the $(i-1)$-th step, it holds that the distribution of the next message
  $M_{i}$ satisfies that $\frac{1}{3}\leq \cProb{}{x,y,T}{M_i = 1}\leq \frac{2}{3}$.
\end{definition}

An important fact that we will use, is that one can transform a given protocol $\Pi$ into a smooth protocol $\Pi'$ that has
roughly the same error probability, whose information cost is the same as the original protocol $\Pi$. Such statement was
proved in~\cite[Lemma 23]{BW} for internal information, and the same argument also works for external information. We thus
have the following lemma.
\begin{lemma}\label{lem:smooth}
  Suppose $f\colon\power{n}\times\power{n}\to\power{}$ and that $\mu$ is a distribution over $\power{n}\times\power{n}$,
  and $\eps,\eps'>0$.
  Then any protocol $\Pi$ for $(f,\mu)$ with external information at most $M$ and error at most $\eps'$
  can be turned into a smooth protocol $\Pi'$ for $(f,\mu)$ external information at most $M$ and error at most $\eps'+\eps$.
\end{lemma}
We remark that the transformation in~\cite[Lemma 23]{BW} may increase the communication complexity of a protocol, however this
is a non-issue for us as we are only concerned with the external information of the protocol.

\begin{fact}\label{fact:KL_abs_val}
  Let $P, Q$ be distributions over domain $X$. Then
  \[
  \sum\limits_{x}{P(x) \card{\log\left(\frac{P(x)}{Q(x)}\right)}}\leq \DKL{P}{Q} + 8.
  \]
\end{fact}
\begin{proof}
  Partition $A = \sett{x}{P(x)\leq Q(X)}$ into $A_j = \sett{x}{2^{-j-1}Q(x)<P(x)\leq 2^{-j} Q(x)}$ where $j=0,1,\ldots$.
  Then the left hand side is
  \[
  \sum\limits_{x\in \bar{A}}{P(x) \log\left(\frac{P(x)}{Q(x)}\right)}
  +\sum\limits_{j=0}^{\infty}\sum\limits_{x\in A_j} P(x) \log\left(\frac{Q(x)}{P(x)}\right)
  =
  \DKL{P}{Q}
  +2\sum\limits_{j=0}^{\infty}\sum\limits_{x\in A_j} P(x) \log\left(\frac{Q(x)}{P(x)}\right).
  \]
  The proof is now concluded by noting that
  \[
  \sum\limits_{j=0}^{\infty}\sum\limits_{x\in A_j} P(x) \log\left(\frac{Q(x)}{P(x)}\right)
  \leq \sum\limits_{j=0}^{\infty}\sum\limits_{x\in A_j} 2^{-j} Q(x)(j+1)
  \leq \sum\limits_{j=0}^{\infty} 2^{-j} (j+1) = 4.\qedhere
  \]
\end{proof}

\subsection{Probability}
We will need the following immediate corollary of Doob's martingale inequality.
\begin{fact}\label{fact:martingale}
  Suppose that $(X_i)_{i=1,\ldots,m}$ is a martingale and $\Expect{}{X_m^2}\leq M$. Then for every $\eps>0$,
  \[
  \Prob{}{\exists i\text{ such that } \card{X_i}\geq \sqrt{M/\eps}}\leq \eps.
  \]
\end{fact}
\begin{proof}
  By Doob's martingale inequality,
  $
  \Prob{}{\max_{i} \card{X_i}\geq \sqrt{M/\eps}}
  \leq \frac{\Expect{}{\card{X_m}}}{M/\eps}
  \leq \frac{\sqrt{\Expect{}{X_m^2}}}{M/\eps}
  \leq \eps$.
\end{proof}

\begin{definition}
  Let $X,Y$ be random variables over the same universe $U$. The statistical distance between $X,Y$ is
  \[
  {\sf SD}(X,Y) = \half\sum\limits_{u\in U}\card{\Prob{}{X=u}-\Prob{}{Y=u}}.
  \]
\end{definition}
\begin{fact}\label{fact:TVD}
  Let $X, Y$ be discrete random variables over the same universe $U$. Then there is
  $A\subseteq U$ such that
  ${\sf SD}(X,Y) = 1 - \sum\limits_{x\in A} \Prob{}{X=x} - \sum\limits_{x\not\in A} \Prob{}{Y=x}$
\end{fact}

\section{Separating internal and external information cost}\label{sec:int_ext}
In this section, we prove Theorem~\ref{thm:external_internal}. A key notion of our proof will
be the relative-discrepancy measure, introduced in~\cite{GKR}.
\begin{definition}\label{def:relative_discrepancy_lot}
  For a function $f\colon\power{n}\times\power{n}\to\power{}$ and a distribution $\mu$ over $\power{n}\times\power{n}$, we say
  $(f,\mu)$ has $(\eps,\delta)$ relative-discrepancy with respect to a distribution $\rho$, if for any rectangle
  $R = A\times B\subseteq \power{n}\times\power{n}$ for which $\rho(R)\geq \delta$, it holds that
  \begin{enumerate}
    \item $\mu(R\cap f^{-1}(0))\geq (\half - \eps)\rho(R)$,
    \item and $\mu(R\cap f^{-1}(1))\geq (\half - \eps)\rho(R)$.
  \end{enumerate}
\end{definition}

\begin{definition}\label{def:relative_discrepancy}
  For a function $f\colon\power{n}\times\power{n}\to\power{}$ and a distribution $\mu$ over $\power{n}\times\power{n}$, we say
  $(f,\mu)$ has $(\eps,\delta)$ relative-discrepancy if there is a distribution $\rho$ such that $(f,\mu)$ has
  $(\eps,\delta)$ relative-discrepancy with respect to $\rho$.
\end{definition}

In \cite{GKR}, the authors show that if $(f,\mu)$ has strong relative-discrepancy,  then $\CC_{\mu}(f,1/2 - \eps')$ must be high.
More precisely, they show if $(f,\mu)$ has $(\eps,\delta)$ discrepancy, then any protocol for $f$ on the distribution $\mu$ that achieves
advantage of $\eps'$, must communicate at least $\log\left(\frac{\eps' - \eps}{\delta}\right)$ bits. The main result of this section
strengthens this assertion, as follows.
\begin{thm}\label{thm:external}
 Let $M\in\mathbb{N}$, $\delta, \eps > 0$ and let $\mu$ be a distribution over $\power{n}\times\power{n}$.
  Suppose $f\colon\power{n}\times\power{n}\to\power{}$ such that $(f,\mu)$ has $(\eps,\delta)$ relative-discrepancy, and
  $\Pi$ is a protocol for computing $f$ such that $\MI^{{\sf ext}}_{\mu}[\Pi] \leq M$, then
  \[
  \Prob{(x,y)\sim\mu}{\Pi(x,y) = f(x,y)}\leq \frac{1}{2} + 2000\max\left(\eps,\frac{M}{\log(1/\delta)}\right).
  \]
\end{thm}
Contrapositively, if $(f,\mu)$ has $(\eps,\delta)$ relative-discrepancy and $\Pi$ is a protocol for $(f,\mu)$ achieving
advantage $\eps'$, then the external information of $\Pi$ according to $\mu$ is at least $(\eps' - 2000\eps)\log(1/\delta)$.
Therefore, the relative-discrepancy measure allows us to prove lower bounds on the external information cost of a function,
which is always smaller than the communication complexity of a function.

To prove Theorem~\ref{thm:external_internal}, we instantiate Theorem~\ref{thm:external} with the ``bursting noise function'' from~\cite{GKR},
which we present next.
\paragraph{The bursting noise function.}
Let $k\in\mathbb{N}$ be thought of as large, and set $c = 2^{4^k}$. The bursting noise function, $f_{\sf burst}$ is a pointer chasing function on a tree of height $c$,
however the input distribution $\mu$ is supported only on $x,y$ that are very correlated. More precisely, for $b\in\{0,1\}$,
we define the distribution $\mu_b$ according to the following sampling procedure:
we think of a complete binary tree of depth $c$, and for each vertex in the tree, each player has a bit in their input. We think of vertices from
odd layers as being owned by Alice, and vertices from even layers as being owned by Bob.
Partition the layers of the tree into $c/k$ multi-layers (a multi-layer consists of $k$ consecutive layers),
and sample $i\in\{1,\ldots, c/k\}$ uniformly. For each multi-layer $j<i$, and for each vertex $u$ in multi-layer $j$,
we choose $y_u\in\power{}$ uniformly, and set $x_u = b\oplus y_u$.
In layer $i$, for each vertex $v$ in it, we choose $x_v,y_v\in\power{}$ independently and uniformly.

Next, we define the notion of a typical vertex. We say a vertex $p$ from layer $i$ is typical, if considering the part of the path from the
root to $p$ that is inside multi-layer $i$, on at least $80\%$ of it, on at least $80\%$ of the odd locations on that path it agrees with $x$,
and on at least $80\%$ of the even locations of the path it agrees with $y$.

For the rest of the layers, for each vertex $u$, let $p(u)$ denote the ancestor of $u$ from layer $i$. If $p(u)$ is typical, we again take
the bits to be uniform such as $x_u = b\oplus y_u$, and if $p(u)$ is atypical we take $x_u,y_u$ as independently chosen bits.

For $(x,y)\in{\sf supp}(\mu_b)$, we define $f_{\sf burst}(x,y) = b$. We take $\mu = \half \mu_0 + \half \mu_1$. We remark that
each one of $x$ and $y$ are $n$-bit Boolean strings where $n = \Theta(2^{c}) = \Theta(2^{2^{4^k}})$.

%

In~\cite{GKR}, the following two important properties are proved for the bursting noise function.
\begin{lemma}\label{lemma:GKR_internal}
  $\IC^{\sf internal}_{\mu}[f,2^{-k}] = O(k)$.
\end{lemma}

\begin{lemma}\label{lem:GKR}
  The pair $(f,\mu)$ has the $(\eps,\delta)$ relative-discrepancy property with respect to $\rho$ for $\eps = 2^{-k}$ and $\delta = \eps/ 2^{2^{k}}$.
\end{lemma}

First, we quickly show that Theorem~\ref{thm:external_internal} follows from Theorem~\ref{thm:external} and the above lemmas.
\begin{proof}[Proof of Theorem~\ref{thm:external_internal}]
Fix $\eta>0$ a small constant, and choose $k$ large enough. Using Lemma~\ref{lem:GKR} together with~\ref{thm:external} gives us that
$\IC^{\sf external}_{\mu}[f,\eta]\geq 2^{\Omega(k)}$, whereas Lemma~\ref{lemma:GKR_internal} gives us that
$\IC^{\sf internal}_{\mu}[f,\eta] = O(k)$.
\end{proof}

The rest of this section is devoted to the proof of Lemma~\ref{thm:external}, and we begin by giving a proof outline.
\paragraph{Outline of the proof of Theorem~\ref{thm:external}.}
The proof has two components. Fix a function $f$ and an input distribution $\mu$. In the first step
we show that any protocol $\Pi$ for $(f,\mu)$ with low external information, can be converted into a protocol $\Pi'$ such that
(a) $\Pi'$ has roughly the same error in computing $(f,\mu)$, and (b) $\Pi'$ never reveals too much information about the player's
input, with respect to any measure $\rho$; we refer to such protocols as having ``universally low external information'' (defined formally below).
In the second step, we show that if $(f,\mu)$ has the $(\eps,\delta)$ relative-discrepancy property, then a protocol
$\Pi$ with low universal external information can only have a small advantage of in computing $(f,\mu)$. Quantitative issues aside,
it is clear that one can combine steps (a) and (b) above to prove that a low external information protocol cannot have a significant
advantage in computing a function that has low relative-discrepancy.

\subsection{Universal external information}
Suppose $\Pi$ is a protocol between Alice and Bob. Suppose Alice speaks first, and denote her messages by $A = (A_1,\ldots, A_m)$, and
Bob's messages by $B = (B_1,\ldots,B_m)$. For each point $i\in [m]$ in the protocol, a possible exchange of messages $(a,b)\in\power{m}\times\power{m}$,
and a pair of inputs $x,y\in\power{k}$, denote
\[
P_{A,\Pi,i}^x(a,b) = \prod\limits_{j< i} \cProb{(X,Y)\sim \mu}{A_{<j} = a_{<j}, B_{<j} = b_{<j}, X=x}{A_j = a_j},
\]
\[
P_{B,\Pi,i}^y(a,b) = \prod\limits_{j< i} \cProb{(X,Y)\sim \mu}{A_{\leq j} = a_{\leq j}, B_{<j} = b_{<j}, Y=y}{B_j = b_j}.
\]
If this product runs through the whole protocol, i.e. $i = m+1$, we omit the subscript $i$ and simply write
$P_{A,\Pi}^x(a,b)$ and $P_{B,\Pi}^y(a,b)$.
\begin{definition}\label{def:low_universal}
  With the above notations, we say a protocol $\Pi$ has universal external information at most $M$,
  if there are non-negative functions $\eta_{A}(a,b)$ and $\eta_{B}(a,b)$ over transcripts, such that the following holds.
  \begin{enumerate}
    \item The function $\eta(a,b) = \eta_A(a,b)\eta_B(a,b)$ is a probability distribution.
    \item For any $(x,y)\in\power{k}\times\power{k}$
  and $(a,b)\in\power{m}\times\power{m}$ it holds that
  \begin{equation}\label{eq:unviversal_small}
  2^{-M}\leq \frac{P_{A,\Pi}^x(a,b) }{\eta_{A}(a,b)}\leq 2^M,
  \qquad \qquad
  2^{-M}\leq \frac{P_{B,\Pi}^y(a,b) }{\eta_{B}(a,b)}\leq 2^M.
  \end{equation}
  \end{enumerate}
\end{definition}
Informally, a protocol has low universal external information, if for all possible transcript $\pi = (a,b)$,
no input of Alice (or Bob) makes $\pi$ much more likely from their point of view.
We remark that having low universal external information is a very strong property. For example, it implies that
the external information of the protocol is low with respect to \emph{any} distribution.

\begin{lemma}\label{lem:low_universal}
  Suppose that a protocol $\Pi$ has universal external information at most $M$. Then
  for any distribution $\rho$ over $(x,y)$, we have that $\MI_{\rho}^{{\sf external}}[\Pi]\leq 2M$.
\end{lemma}
\begin{proof}
Let $\eta_A, \eta_B$ and $\eta = \eta_A\cdot \eta_B$ be from Definition~\ref{def:low_universal}. First, we argue that
for all inputs $x,y$ it holds that $\DKL{\Pi|_{X=x, Y=y}}{\eta}\leq 2M$. Indeed, by definition
\[
  \DKL{\Pi|_{X=x, Y=y}}{\eta}
  =\sum\limits_{\pi} \cProb{}{X=x, Y=y}{\Pi = \pi} \log\left(\frac{\cProb{}{X=x, Y=y}{\Pi = \pi}}{\eta(\pi)}\right).
\]
Noting that $\cProb{}{X=x, Y=y}{\Pi = \pi} =  P_{A,\Pi}^x(\pi)P_{B,\Pi}^y(\pi)$, we get from the universal external information property that
\[
\frac{\cProb{}{X=x, Y=y}{\Pi = \pi}}{\eta(\pi)}\leq 2^{2M},
\]
and plugging that in above yields that $\DKL{\Pi|_{X=x, Y=y}}{\eta}\leq 2M$.

The statement will thus follow if we show that $\MI^{{\sf external}}_{\rho}[\Pi]\leq \Expect{(x,y)\sim \rho}{\DKL{\Pi|_{X=x, Y=y}}{\eta}}$.
Indeed, using the definition of external information and Fact~\ref{fact:mutual_div_KL} we get that
  \[
  \MI^{{\sf external}}_{\rho}[\Pi]
  =\MI_{\rho}[X,Y;\Pi]
  =\Expect{(x,y)\sim\rho}{\DKL{\Pi|_{X=x,Y=y}}{\Pi}},
  \]
  and therefore
  \begin{align*}
  \Expect{(x,y)\sim \rho}{\DKL{\Pi|_{X=x, Y=y}}{\eta}} - \MI^{{\sf external}}_{\rho}[\Pi]
  &=\Expect{(x,y)\sim\rho}{\DKL{\Pi|_{X=x,Y=y}}{\eta} - \DKL{\Pi|_{X=x,Y=y}}{\Pi}}\\
  &=\Expect{(x,y)\sim\rho}{\sum\limits_{\pi}\Prob{}{\Pi(x,y) = \pi}\log\left(\frac{\Prob{X,Y\sim\rho}{\Pi = \pi}}{\eta(\pi)}\right)}\\
  &=\sum\limits_{\pi}\Prob{X,Y\sim\rho}{\Pi = \pi}\log\left(\frac{\Prob{X,Y\sim\rho}{\Pi = \pi}}{\eta(\pi)}\right)\\
  &= \DKL{\Pi}{\eta} \geq 0.\qedhere
  \end{align*}

\end{proof}

\subsection{Step (a): fixing external information leakage}\label{sec:fixing_leak}
Our goal in this section is to prove the following lemma, asserting that a low external information protocol may be converted into
a protocol with low universal external information with only small additional error.
\begin{lemma}\label{lem:convert_low_universal}
  Let $\eps,\eps' > 0$ and let $\mu$ be distributions over $\power{n}\times\power{n}$.
  Suppose $f\colon\power{n}\times\power{n}\to\power{}$ is a function, and $\Pi$ is a protocol for $(f,\mu)$
  that has error at most $\eps'$ and $\MI^{{\sf external}}_{\mu}[\Pi] \leq M$. Then there is a protocol $\Pi'$ for $(f,\mu)$ such that
  \begin{enumerate}
    \item The error of $\Pi'$ on $(f,\mu)$ is at most $\eps'+40\eps$.
    \item The universal external information of $\Pi'$ is at most $2M/\eps+1$.
  \end{enumerate}
\end{lemma}

We begin by explaining the idea in behind the design of $\Pi'$.
In $\Pi'$ we will simulate $\Pi$, except that each player will also measure how many bits of external information
they have leaked so far (this is possible to do since it only depends on the transcript up to that point, and their input). In case
this number of bits has exceeded a certain threshold, the player changes their behaviour and enters a ``strike'' in which they will
act in a way that does not reveal any additional external information regarding his/her input. Strictly speaking, once a player determines
they have leaked too much information, they will forget about their input and instead sample their message according to their message
distribution at that point in $\Pi$, conditioned on the transcript so far (but not on their specific input).

Let us now be more precise. Fix a protocol $\Pi$ such that $I^{\sf ext}_\mu[\Pi]\leq M$, and recall the definitions of $P^x_{A,\Pi,i}(a,b)$ and
$P^y_{B,\Pi,i}(a,b)$ above. We will also need to define their averaged counterparts, i.e.
\[
\eta_{A,i}(a,b) \defeq P_{A,\Pi,i}(a,b) = \prod\limits_{j< i} \cProb{(X,Y)\sim \mu}{A_{<j} = a_{<j}, B_{<j} = b_{<j}}{A_j = a_j},
\]
\[
\eta_{B,i}(a,b) \defeq P_{B,\Pi,i}(a,b) = \prod\limits_{j< i} \cProb{(X,Y)\sim \mu}{A_{\leq j} = a_{\leq j}, B_{<j} = b_{<j}}{B_j = b_j},
\]
(the $\eta_{A,i},\eta_{B,i}$ notations is not a coincidence, and we will use these functions to exhibit the fact that the protocol $\Pi'$ we
construct has low universal external information).
We note that $P_{A,\Pi,i}, P^{x}_{A,\Pi,i}$ only depend on the $i-1$-prefixes of $a$ and $b$, and $P_{B,\Pi,i}, P^{y}_{B,\Pi,i}$ only depends on the
$i$-prefix of $a$ and $i-1$-prefix of $b$. We will therefore sometimes abuse notations and drop the rest of $a,b$ from the notation.
We also note that $P^{x}_{A,\Pi,i}, P^{y}_{A,\Pi,i}$ depend only on $x$ and $y$ (and not on $\mu$). This is because, at each point in time, a player's message only depends
on their input, and the messages they received from the other player so far.

With these notations, we may consider for each $a,b$, the likelihood ratios
$S_{A,i}(a_{<i},b_{<i},x) = \frac{P^{x}_{A,i}(a_{<i},b_{<i})}{\eta_{A,i}(a,b)}$ and
$S_{B,i}(a_{\leq i},b_{< i},y) = \frac{P^{y}_{B,i}(a_{\leq i},b_{< i})}{\eta_{B,i}(a,b)}$. Intuitively,
these quantities measure how much more/ less likely
a given exchange of messages $(a,b)$ is, when knowing $x$ and $y$ respectively, compared to only knowing that $(x,y)\sim\mu$.
Thus, we may expect an external observer to learn many bits of information in case the protocol was executed and the resulting exchange of messages $a,b$ has
high likelihood ratios, say $S_{A,m}(a,b,x)\geq 2^M$ (in which case we expect an external observer to learn $\approx M$ bits of information).
This turns out to be true, and actually with slightly more work, one can show that the same holds if the likelihood ratios become large at some earlier point in the protocol, $i<m$.

With this intuition in mind, and noting that Alice (and analogously Bob) can compute $S_{A,i}(a_{<i},b_{<i},x)$ (analogously $S_{B,i}(a_{\leq i},b_{< i},y)$)
it makes sense that the players should alter their behaviour if at some point in their protocol, their likelihood ratio gets too high -- say,
larger than $2^{2 M/\eps}$. Indeed, this is what our protocol $\Pi'$ does.

\paragraph{The protocol $\Pi'$.} We simulate the protocol $\Pi$, with a small change in the beginning of each player's turn.
Consider a player intending to send their $i$th message -- say Alice. First,
Alice  computes $S_{A,i}(\pi_A,\pi_B,x)$ (where $\pi_A$ are the messages of Alice so far, and $\pi_B$ are the messages of Bob so far).
If this quantity is larger than $2^{2 M/\eps}$, or at most $2^{-2M/\eps}$, Alice moves into \emph{``strike mode''}, and otherwise proceeds
as usual according to the protocol $\Pi$.
Upon entering ``strike mode'', Alice will sample her subsequent messages only conditioned on the transcript of the protocol up to that point
without taking her input $x$ into consideration.
I.e., to send her $j$th message, for $j\geq i$, Alice considers the transcript of the protocol thus far, $a_{<j}, b_{<j}$,
and the distribution $A_j(X,Y)~|~A_{<j}(X,Y) = a_{<j}, B_{<j}(X,Y) = a_{<j}$ where $(X,Y) \sim \mu$,
and samples her next message according to it. Bob implements an analogous check during his turns.

\skipi
In the remainder of this section, we argue that the probability that a player ever enters ``strike mode'' in $\Pi'$ is small, and so
$\Pi'$ retains roughly the same error as $\Pi$ on $(f,\mu)$. We then show that  $\Pi'$ has universal external information at most $2M/\eps+1$.
We remark that for technical reasons, we will need to assume that our original protocol is smooth (as in Definition~\ref{def:smooth}). Thankfully, by Lemma~\ref{lem:smooth},
we may indeed do so while only slightly increasing the error of the protocol.

\subsubsection{The error of $\Pi'$ on $\mu$ is comparable to the error of $\Pi$}
In this section we prove the following lemma.
\begin{lemma}\label{lem:advantage_same}
  The probability that at least one of the players enters ``strike mode'' in the protocol $\Pi'$ when ran on $\mu$ is at most $38\eps$.
\end{lemma}
First, by Lemma~\ref{lem:smooth} we may assume henceforth that the protocol $\Pi$ is smooth and has error at most $\eps'+\eps$.
Thus, once we prove Lemma~\ref{lem:advantage_same} it will follow that the error of $\Pi'$
is at most $\eps'+39\eps$. The rest of this section is therefore devoted to the proof of Lemma~\ref{lem:advantage_same}.

By Fact~\ref{fact:mutual_div_KL}
and the definition of KL-divergence
\begin{align}
M
\geq \MI_{\mu}[X,Y; \Pi]
&= \Expect{(x,y)\sim \mu}{\DKL{\Pi|_{X=x,Y=y}}{\Pi}} \notag\\\label{eq:1}
&= \Expect{(x,y)\sim\mu}{\sum\limits_{a,b} \Prob{}{\Pi(x,y) = (a,b)} \log\left(\frac{\Prob{}{\Pi(x,y) = (a,b)}}{\Prob{}{\Pi = (a,b)}}\right)}.
\end{align}
For $a,b\in\power{m}$, we define
$p_{A,\Pi,i}(a,b) = \cProb{(X,Y)\sim\mu}{A_{<i} = a_{<i}, B_{<i} = b_{<i}}{ A_i = a_i}$, and similarly
we define for Bob $p_{B,\Pi,i}(a,b) = \cProb{(X,Y)\sim\mu}{A_{\leq i} = a_{\leq i}, B_{<i} = b_{<i}}{ B_i = b_i}$.
Also, let
\[
p^x_{A,\Pi,i}(a,b) = \cProb{(X,Y)\sim\mu}{A_{<i} = a_{<i}, B_{<i} = b_{<i}, X = x}{ A_i = a_i}
\]
and
\[
p^{y}_{B,\Pi,i}(a,b) = \cProb{(X,Y)\sim\mu}{A_{\leq i} = a_{\leq i}, B_{<i} = b_{<i}, Y = y}{ B_i = b_i}.
\]
We remark that $p^x_{A,\Pi,i}(a,b), p_{A,\Pi,i}(a,b)$ depend only on the $i$-prefix of $a$ and the $i-1$ prefix of
$b$, and $p^{y}_{B,\Pi,i}(a,b), p_{B,\Pi,i}(a,b)$ depend on the $i$-prefixes of both $a$ and $b$. Thus, abusing notations,
we sometimes plug in strings of length $i$ into $p_{B,\Pi,i}(a,b)$ and so on.

\skipi
With these notations, we have
\[
\Prob{}{\Pi = (a,b)} = \prod\limits_{i=1}^{m} p_{A,\Pi,i}(a,b)\prod\limits_{i=1}^{m} p_{B,\Pi,i}(a,b),
\]
and for every fixed $x,y$ it holds that $\Prob{}{\Pi(x,y) = (a,b)} = \prod\limits_{i=1}^{m} p^x_{A,\Pi,i}(a,b)\prod\limits_{i=1}^{m} p^{y}_{B,\Pi,i}(a,b)$.
Thus, ~\eqref{eq:1} gives us that
\begin{equation}\label{eq:0}
\sum\limits_{i=1}^{m} \Expect{(x,y)\sim \mu}{\sum\limits_{a,b}
{\Prob{}{\Pi(x,y) = (a,b)}
\left(\log\left(\frac{p^x_{A,\Pi,i}(a,b)}{p_{A,\Pi,i}(a,b)}\right)
+\log\left(\frac{p^y_{B,\Pi,i}(a,b)}{p_{B,\Pi,i}(a,b)}\right)\right)}}
\leq M.
\end{equation}
We consider the two terms on the left hand side separately, i.e. define
\begin{equation}\label{eq:2}
(\rom{1})
=\sum\limits_{i=1}^{m} \Expect{(x,y)\sim \mu}{\sum\limits_{a,b}
{\Prob{}{\Pi(x,y) = (a,b)} \log\left(\frac{p^x_{A,\Pi,i}(a,b)}{p_{A,\Pi,i}(a,b)}\right)}},
\end{equation}
\begin{equation}\label{eq:3}
(\rom{2})
=\sum\limits_{i=1}^{m} \Expect{(x,y)\sim \mu}{\sum\limits_{a,b}
{\Prob{}{\Pi(x,y) = (a,b)} \log\left(\frac{p^y_{B,\Pi,i}(a,b)}{p_{B,\Pi,i}(a,b)}\right)}}.
\end{equation}
We now take a moment to reinterpret these two quantities.
Define $Z_{x,a,b,A,i}, Z_{y,a,b,B,i}\colon \{0,1\}\to \mathbb{R}$ as
\[
Z_{x,a,b,A,i}(c) = \log\left(\frac{p^x_{A,\Pi,i}(a_{<i},c , b_{\leq i})}{p_{A,\Pi,i}(a_{<i}, c, b_{\leq i})}\right),
\qquad
Z_{y,a,b,B,i}(c) = \log\left(\frac{p^y_{A,\Pi,i}(a_{\leq i},b_{<i}, c)}{p_{B,\Pi,i}(a_{\leq i},b_{<i}, c)}\right).
\]
Here, $c$ is to be thought of as the next message of the respective player, conditioned on their input and the transcript so far.
Note that the distribution of $Z_{x,a,b,A,i}$, only depends on the $(i-1)$ prefix of $a,b$ and the distribution of $Z_{x,a,b,B,i}$
only depends on the $i$ prefix of $a$ and $(i-1)$-prefix of $b$.
Let $E_{x,a,b,A,i}, E_{y,a,b,B,i}$ be
their expectations, respectively, i.e.
\[
E_{x,a,b,A,i} =
\sum\limits_{c\in\power{}}
{
p^x_{A,\Pi,i}(a_{<i}, c,b_{\leq i})
Z_{x,a,b,A,i}(c)},
\qquad
E_{x,a,b,B,i} =
\sum\limits_{c\in\power{}}
{
p^y_{B,\Pi,i}(a_{\leq i}, b_{<i}, c)
Z_{x,a,b,B,i}(c)}.
\]
With these notations, we note that~\eqref{eq:2} and~\eqref{eq:3} translate to
\begin{equation}\label{eq:4}
(\rom{1}) =
\sum\limits_{i=1}^{m}
\Expect{(x,y)\sim \mu}{\Expect{\substack{a_{<i}\sim A_{<i}(x,y)\\ b_{<i}\sim B_{<i}(x,y)}}
{E_{x,a,b,A,i}}},
\qquad
(\rom{2}) =
\sum\limits_{i=1}^{m}
\Expect{(x,y)\sim \mu}{\Expect{\substack{a_{\leq i}\sim A_{\leq i}(x,y)\\ b_{<i}\sim B_{<i}(x,y)}}
{E_{y,a,b,B,i}}}.
\end{equation}

\paragraph{Analyzing the probability to enter ``strike mode''.}
With the notations we have set, we have that
\[
\log(S_{A,i}(a,b,x)) = \sum\limits_{j< i} Z_{x,a,b,A,i}(a_i),
\]
and similarly for Bob, and so probability that one of the players in $\Pi'$ enters ``strike mode'' is
\[
\Prob{\substack{(x,y)\sim\mu\\ a\sim A,b\sim B}}{
\exists i\in [m]\text{ such that }\card{\sum\limits_{j < i} Z_{x,a,b,A,i}(a_i)}\geq\frac{2M}{\eps}
\text{ or }\card{\sum\limits_{j < i} Z_{y,a,b,B,i}(b_i)}\geq \frac{2M}{\eps}}.
\]
The rest of the proof is dedicated to upper bounding the probability that
$\card{\sum\limits_{j\leq i} Z_{x,a,b,A,i}(a_i)} \geq \frac{2M}{\eps}$,
and the probability that $\card{\sum\limits_{j< i} Z_{x,a,b,B,i}(b_i)}\geq \frac{2M}{\eps}$,
each by $19\eps$. Lemma~\ref{lem:advantage_same} thus follows from the union bound. We focus on
upper bounding the probability for Alice, and the argument for Bob is analogous.

\begin{claim}\label{claim:pinsker_conclusion}
For each $x,y,a,b$ and $i$, we have that $E_{x,a,b,A,i},E_{x,a,b,B,i}$ are non-negative and furthermore
\begin{align*}
&E_{x,a,b,A,i}\geq 2(
p^x_{A,\Pi,i}(a_{<i}, 1,b_{\leq i})
-p_{A,\Pi,i}(a_{<i}, 1,b_{\leq i}))^2,\\
&E_{y,a,b,B,i}\geq
2(
p^y_{B,\Pi,i}(a_{\leq i},b_{< i}, 1)
-
p_{B,\Pi,i}(a_{\leq i},b_{< i}, 1))^2.
\end{align*}
\end{claim}
\begin{proof}
  We show the argument for $E_{x,a,b,A,i}$, and the argument for $E_{x,a,b,B,i}$ is identical. Note that
  \[
  E_{x,a,b,A,i} = \DKL{A_i|_{X=x, A_{<i} = a_{<i}, B_{<i} = b_{<i}}}{A_i|_{A_{<i} = a_{<i}, B_{<i} = b_{<i}}},
  \]
  from which the non-negativity is clear.  Since $\Pi$ is a smooth protocol, we may use Fact~\ref{fact:KL_bits} and conclude that
  \[
  E_{x,a,b,A,i}\geq 2(p^x_{A,\Pi,i}(a_{<i}, 1,b_{\leq i})
-p_{A,\Pi,i}(a_{<i}, 1,b_{\leq i}))^2.\qedhere
  \]
\end{proof}

We note that Claim~\ref{claim:pinsker_conclusion} combined with~\eqref{eq:0} and~\eqref{eq:4} immediately implies:
\begin{corollary}\label{corr:immediate}
We have that
\[
  M\geq (\rom{1})\geq
  2\sum\limits_{i=1}^{m}
  \Expect{(x,y)\sim \mu}{
  \Expect{\substack{a_{<i}\sim A_{<i}(x,y)\\ b_{<i}\sim B_{<i}(x,y)}}
  {(p^x_{A,\Pi,i}(a_{<i}, 1,b_{\leq i})
-p_{A,\Pi,i}(a_{<i}, 1,b_{\leq i}))^2}}.
\]
\end{corollary}
\skipi

Consider a random choice of $(x,y)\sim\mu$, $a\sim A(x,y)$ and $b\sim B(x,y)$.
Note that the sequence $Q^A_i = Z_{x,a,b,A,i}(a_i) - E_{x,a,b,A,i}$ forms the sum-martingale
$G^A_i = \sum\limits_{j\leq i} Q^{A}_j$, and the similarly the sequence $Q^B_i = Z_{x,a,b,B,i}(b_i) - E_{x,a,b,B,i}$
forms the sum martingale $G^B_i = \sum\limits_{j\leq i} Q^{B}_j$, both with respect to the natural filtration defined by the transcript
of the protocol at each step. The following claim upper bounds the expectation of the square of these sum-martingales in the end.

\begin{claim}\label{claim:expected_sum_square}
  $\Expect{}{(G^A_m)^2}\leq 18(\rom{1})\leq 18 M$.
\end{claim}
\begin{proof}
Using the martingale property we have that
\[
\Expect{}{(G^A_m)^2}
=\sum\limits_{i=1}^m \Expect{}{(Q^{A}_i)^2}
\leq \sum\limits_{i=1}^m \Expect{}{Z_{x,a,b,A,i}(a_i)^2}.
\]
Note that fixing $x, a_{<i}, b_{<i}$ we have that
\begin{align*}
\Expect{a_i}{Z_{x,a,b,A,i}(a_i)^2}
&=\sum\limits_{c\in\power{}} p^x_{A,\Pi,i}(a_{<i}, c, b_{<i}) \log^2\left(\frac{p^x_{A,\Pi,i}(a_{<i}, c, b_{<i})}{p_{A,\Pi,i}(a_{<i}, c, b_{<i})}\right).
\end{align*}
Using the smoothness of $\Pi$, we have that $\frac{p^x_{A,\Pi,i}(a_{<i}, c, b_{<i})}{p_{A,\Pi,i}(a_{<i}, c, b_{<i})}$ is at least
$\half$ for all $c\in\power{}$, and since $\card{\log(z)}\leq 2\card{z-1}$ for all $z\geq 1/2$ we get that
\[
\Expect{a_i}{Z_{x,a,b,A,i}(a_i)^2}
\leq 4\sum\limits_{c\in\power{}} p^A_{a_{<i},b_{<i},i}(x,c)
\left(
\frac{p^x_{A,\Pi,i}(a_{<i}, c, b_{<i})-p_{A,\Pi,i}(a_{<i}, c, b_{<i})}{p_{A,\Pi,i}(a_{<i}, c, b_{<i})}\right)^2.
\]
By the smoothness of $\Pi$ we have $p_{A,\Pi,i}(a_{<i}, c, b_{<i})\geq 1/3$ for all $c\in\power{}$, so the above
inequality implies that
$\Expect{a_i}{Z_{x,a,b,A,i}(a_i)^2}\leq 36\left(pp^x_{A,\Pi,i}(a_{<i}, 1, b_{<i})-p_{A,\Pi,i}(a_{<i}, 1, b_{<i})\right)^2$.
The claim now follows by combining this with Corollary~\ref{corr:immediate}.
\end{proof}

We are now ready to prove Lemma~\ref{lem:advantage_same}.
\begin{proof}[Proof of Lemma~\ref{lem:advantage_same}]
Consider a random choice of $(x,y)\sim\mu$, $a\sim A(x,y)$ and $b\sim B(x,y)$. Let
$W_1$ be the event that $\sum\limits_{i=1}^{m} E_{x,a,b,A,i}\geq M/\eps$ and let
$W_2$ be the event that $\card{G^A_i}\geq \sqrt{M/\eps}$ for some $i$. By Markov's inequality we have that
\[
\Prob{}{W_1}\leq \frac{(\rom{1})}{M/\eps}\leq \eps,
\]
where we used Corollary~\ref{corr:immediate}. For $W_2$, using Fact~\ref{fact:martingale} and Claim~\ref{claim:expected_sum_square} gives that
\[
\Prob{}{W_2}\leq \frac{\Expect{}{(G^A_m)^2}}{M/\eps}\leq 18\eps.
\]
Note that $\card{\sum\limits_{j< k} Z_{x,a,b,A,i}(a_i)} = \card{G^A_{k-1}} + \sum\limits_{i=1}^{m} E_{x,a,b,A,i}$, so if none of $W_1,W_2$
hold, then $\card{\sum\limits_{j< i} Z_{x,a,b,A,i}(a_i)}\leq M/\eps + \sqrt{M/\eps} \leq 2M/\eps$ and Alice never enters ``strike mode''
in the duration on the execution of $\Pi'$. Therefore, the probability Alice enters into ``strike mode'' on an execution of $\Pi'$
is at most $\Prob{}{W_1} + \Prob{}{W_2}\leq 19\eps$. The same goes for Bob, and Lemma~\ref{lem:advantage_same} follows from the union bound.
\end{proof}

\subsubsection{The protocol $\Pi'$ has universal external information at most $2M/\eps + 1$}
Let $\eta_A = \eta_{A,m+1}$ and $\eta_B = \eta_{B,m+1}$.
It is easy to see that the $\eta(a,b) = \eta_{A}(a,b)\eta_{B}(a,b)$ is a distribution.
We show that it exhibits that $\Pi'$ has low universal external information.

Suppose towards contradiction that this is not the case. Then there are inputs $x,y$ and possible transcripts $a,b$
and a step $i$ such that~\eqref{eq:unviversal_small} fails -- suppose without loss of generality that
$\frac{P_{A,\Pi'}^x(a,b)}{\eta_A(a,b)} > 2^{2M/\eps + 1}$.

We claim that on an execution of $\Pi'$ on $x,y$ that
yields the transcript $(a,b)$, it must be the case that Alice entered ``strike mode''. Otherwise,
we have that:
\[
\frac{P_{A,\Pi'}^x(a,b)}{\eta_A(a,b)}
= \frac{P_{A,\Pi,m}^x(a,b)}{\eta_{A,m}(a,b)}\cdot
\frac{\cProb{}{a_{<m}, b_{<m}, x}{A_m = a_m\text{ in $\Pi$}}}{\cProb{}{a_{<m}, b_{<m}}{A_m = a_m\text{ in $\Pi$}}}.
\]
Since Alice did not enter ``strike mode'', the first fraction is between $2^{-2M/\eps}$ and $2^{2M/\eps}$, and since
$\Pi$ is smooth, the second fraction is between $1/2$ and $2$. It follows that~\eqref{eq:unviversal_small} holds
for $2M/\eps + 1$ for Alice, in contradiction.

Let $i$ be the step in the protocol in which Alice decided to enter ``strike mode''.
By the minimality of $i$ and smoothness of $\Pi$ we conclude that
\[
S_{A,i}(a_{<i},b_{<i},x) = S_{A,i-1}(a_{<i-1},b_{<i-1},x) \frac{\cProb{}{a_{<i-1}, b_{<i-1}, x}{A_{i-1} = a_{i-1}\text{ in $\Pi$}}}{\cProb{}{a_{<i-1}, b_{<i-1}}{A_{i-1} = a_{i-1}\text{ in $\Pi$}}}
\leq 2^{2M/\eps}\cdot 2.
\]
Now, note that by the behaviour of Alice in strike mode it follows that
  \[
  P_{A,\Pi'}^x(a,b) = P_{A,\Pi,i}^x(a,b)\cdot \prod\limits_{k=i}^{m}\cProb{(X,Y)\sim \mu}{a_{<k}, b_{<k}}{A_{k} = a_{k}\text{ in $\Pi$}},
  \]
  and clearly
  \[
  \eta_{A,m+1}(a,b) = \eta_{A,i}(a,b)\prod\limits_{k=i}^{m} \cProb{(X,Y)\sim \mu}{a_{<k}, b_{<k}}{A_{k} = a_{k}\text{ in $\Pi$}},
  \]
  so
  \[
  \frac{P_{A,\Pi'}^x(a,b)}{\eta_{A,m}(a,b)}
  =S_{A,i}(a_{<i},b_{<i},x)
  \leq 2^{2M/\eps+1},
  \]
  which is a contradiction.\qed
\subsection{Step (b): strong relative discrepancy implies high universal external information}
Next, we prove the following lemma, asserting that a protocol with low universal external information cannot compute functions
that have low relative-discrepancy.
\begin{lemma}\label{lem:step2}
  Let $\delta, \eps,\eps' > 0$ and $M\in\mathbb{N}$.
  Let $\mu$ be a distribution over $\power{n}\times\power{n}$, and
  suppose that $f\colon\power{n}\times\power{n}\to\power{}$ is such that $(f,\mu)$ has $(\eps,\delta)$ relative-discrepancy.
  If $\Pi$ is a protocol for $(f,\mu)$ whose universal external information at most
  $M$, then for $(X,Y)\sim \mu$,
  \[
  {\sf SD}(\Pi(X,Y)|_{f(X,Y) = 0}, \Pi(X,Y)|_{f(X,Y) = 1})\leq 20\left(\eps+\eps'+\frac{2^{4M}}{\eps'^2}\delta\right).
  \]
\end{lemma}
\begin{proof}
  By Definition~\ref{def:low_universal}, there are functions $\eta_A(a,b),\eta_B(a,b)\geq 0$ such that
  $\eta(a,b) = \eta_A(a,b)\eta_B(a,b)$ is a distribution, and inequality~\eqref{eq:unviversal_small} holds for all $x,y,a,b$. For each
  possible transcript $\pi$, we partition
  $x,y$ into rectangles according to the ratios $\frac{P_{A}^{x}(\pi)}{\eta_{B}(\pi)}$ and $\frac{P_{B}^{y}(\pi)}{\eta_{B}(\pi)}$.
  Namely, for $-M/\eps'\leq i,j< M/\eps'$ we denote
  \[
  R_{\pi}^X[i] = \sett{x}{(1+\eps')^{i}\leq \frac{P_{A}^{x}(\pi)}{\eta_{A}(\pi)}\leq (1+\eps')^{i+1}},
  ~~~
  R_{\pi}^Y[j] = \sett{y}{(1+\eps')^{j}\leq \frac{P_{B}^{y}(\pi)}{\eta_{B}(\pi)}\leq (1+\eps')^{j+1}},
  \]
  and define $R_{\pi}[i,j] = R_{\pi}^X[i]\times R_{\pi}^Y[j]$.  We note that the number of rectangles is at most
  $(2M/\eps')^2$, and that they partition the entire domain.

  Let $\rho$ be a distribution from Definition~\ref{def:relative_discrepancy} exhibiting the fact that $(f,\mu)$ has $(\eps,\delta)$ relative-discrepancy.
  We say a rectangle is heavy if
  $\rho(R_{\pi}[i,j])\geq \delta$ and otherwise we say it is light. For future reference, note that the total $\rho$-weight on light rectangles
  is at most $(2M/\eps')^2 \delta$.

  Fix $\pi$ and let $H_{\pi} = \sett{(i,j)}{ R_{\pi}[i,j]\text{ is heavy}}$.
  Then for all $b\in\power{}$ we have
  \begin{align*}
    \Prob{(x,y)\sim \mu}{ \Pi(x,y) = \pi, f(x,y) = b}
    &=\sum\limits_{\substack{(x,y)\\ f(x,y) = b}}{\mu(x,y) P_{A}^{x}(\pi)P_{B}^{x}(\pi)}\\
    &\geq \sum\limits_{(i,j)\in H_{\pi}}\sum\limits_{\substack{(x,y)\in R_{\pi}[i,j]\\ f(x,y) = b}}{\mu(x,y) P_{A}^{x}(\pi)P_{B}^{x}(\pi)}\\
    &\geq \sum\limits_{(i,j)\in H_{\pi}}\sum\limits_{\substack{(x,y)\in R_{\pi}[i,j]\\ f(x,y) = b}}{(1+\eps')^{i+j}\mu(x,y) \eta_A(\pi)\eta_B(\pi)},
  \end{align*}
  where in the last inequality we used the definition of $R_{\pi}[i,j]$. Thus, we get that
  \[
  \Prob{(x,y)\sim \mu}{ \Pi(x,y) = \pi, f(x,y) = b} \geq \eta(\pi)\sum\limits_{(i,j)\in H_{\pi}} (1+\eps')^{i+j}\mu(R_{\pi}[i,j]\cap f^{-1}(b)).
  \]
  By the relative-discrepancy property we have that $\mu(R_{\pi}[i,j]\cap f^{-1}(b))\geq \left(\half - \eps\right)\rho(R_{\pi}[i,j])$, and so
  \begin{equation}\label{eq:5}
  \Prob{(x,y)\sim \mu}{ \Pi(x,y) = \pi, f(x,y) = b}
  \geq
  \left(\half - \eps\right)\eta(\pi)\sum\limits_{(i,j)\in H_{\pi}} (1+\eps')^{i+j}\rho(R_{\pi}[i,j]),
  \end{equation}
  and we analyze the last sum. Note that
  \begin{align}\label{eq:6}
  \sum\limits_{(i,j)\in H_{\pi}} (1+\eps')^{i+j}\rho(R_{\pi}[i,j])
  &=\sum\limits_{(i,j)\in H_{\pi}}\sum\limits_{(x,y)\in R_{\pi}[i,j]}  (1+\eps')^{i+j}\rho(x,y)\notag\\
  &\geq\sum\limits_{(i,j)\in H_{\pi}} \sum\limits_{(x,y)\in R_{\pi}[i,j]}(1+\eps')^{-2}\frac{P_A^x(\pi)P_B^y(\pi)}{\eta_A(\pi)\eta_B(\pi)}\rho(x,y)\notag\\
  &=\sum\limits_{(i,j)\in H_{\pi}} \sum\limits_{(x,y)\in R_{\pi}[i,j]}(1+\eps')^{-2}\frac{1}{\eta(\pi)}\rho(x,y)P_A^x(\pi)P_B^y(\pi)\notag\\
  &=\frac{(1+\eps')^{-2}}{\eta(\pi)} \Prob{(X,Y)\sim \rho}{\Pi(X,Y) = \pi, (X,Y) \text{ in a heavy rectangle of }\pi}\notag\\
  &=\frac{(1+\eps')^{-2}}{\eta(\pi)} \left(\rho(\pi) - \Prob{(X,Y)\sim \rho}{\Pi(X,Y) = \pi, (X,Y) \text{ in light rectangle}}\right).
  \end{align}
  We now upper bound the last probability. By conditioning we have that
  \[
    \Prob{(X,Y)\sim \rho}{\Pi(X,Y) = \pi, (X,Y) \text{ in a light rectangle of }\pi}
    =\sum\limits_{(i,j)\not\in H_\pi}\sum\limits_{(x,y)\in R_{\pi}[i,j]} \rho(x,y) \rho|_{x,y}(\pi),
  \]
  and by our earlier notations we have
  $\rho|_{x,y}(\pi) = P_A^x(\pi) P_B^y(\pi)\leq 2^{2M} \eta_A(\pi)\eta_B(\pi) = 2^{2M} \eta(\pi)$, so
  \[
  \Prob{(X,Y)\sim \rho}{\Pi(X,Y) = \pi, (X,Y) \text{ in a light rectangle of }\pi}
  \leq 2^{2M} \eta(\pi) \sum\limits_{(i,j)\not\in H_\pi}\sum\limits_{(x,y)\in R_{\pi}[i,j]} \rho(x,y),
  \]
  which is at most $2^{2M} \eta(\pi)\cdot (2M/\eps')^2 \delta \leq \frac{2^{4M}}{\eps'^2}\delta\eta(\pi)$.
  Plugging this into~\eqref{eq:6}, and then~\eqref{eq:6} into~\eqref{eq:5} yields that
  \begin{equation}\label{eq:7}
  \Prob{(x,y)\sim \mu}{ \Pi(x,y) = \pi, f(x,y) = b}
  \geq \left(\half - \eps\right)(1+\eps')^{-2}\left(\rho(\pi) - \frac{2^{4M}}{\eps'^2} \delta  \eta(\pi)\right).
  \end{equation}

  We note that summing this up over all $\pi$, we get that
  \begin{equation}\label{eq:9}
  \Prob{(x,y)\sim \mu}{f(x,y) = b}\geq \left(\half - \eps\right)(1+\eps')^{-2}(1-\frac{2^{4M}}{\eps'^2}\delta)
  \geq \half - \eps -\eps' - \frac{2^{4M}}{\eps'^2}\delta.
  \end{equation}
  We can now bound the statistical distance between $\Pi(X,Y)|_{f(X,Y) = 1}$ and
  $\Pi(X,Y)|_{f(X,Y) = 0}$ where $(X,Y)\sim \mu$. By Fact~\ref{fact:TVD}, there is $A\subseteq {\sf Supp}(\Pi)$ such that
  this statistical distance is equal to
  \[
  1
  -
  \left(\sum\limits_{\pi\in A}
  \cProb{}{f(X,Y) = 1}{\Pi(X,Y) = \pi}
  +\sum\limits_{\pi\not\in A}
  \cProb{}{f(X,Y) = 0}{\Pi(X,Y) = \pi}\right).
  \]
  Denote $p =\Prob{(X,Y)\sim \mu}{f(X,Y) = 1}$, and note that by~\eqref{eq:9} we have that
  $\card{p - \half}\leq \eps + \eps' + \frac{2^{4M}}{\eps'^2}\delta$. We get from the above that the statistical distance is
  equal to
  \begin{align}\label{eq:8}
  &2\left(
  \half -
  \frac{1}{2p}\sum\limits_{\pi\in A}
  \Prob{}{f(X,Y) = 1, \Pi(X,Y) = \pi}
  -\frac{1}{2(1-p)}\sum\limits_{\pi\not\in A}
  \Prob{}{f(X,Y) = 0, \Pi(X,Y) = \pi}
  \right) \notag\\
  \leq
  &2\left(
  \half -
  \sum\limits_{\pi\in A}
  \Prob{}{f(X,Y) = 1, \Pi(X,Y) = \pi}
  -\sum\limits_{\pi\not\in A}
  \Prob{}{f(X,Y) = 0, \Pi(X,Y) = \pi}
  \right)\notag\\
  &+8(\eps + \eps' + \frac{2^{4M}}{\eps'^2}\delta),
  \end{align}
  and it is enough to lower bound the two sums. Define $b_{\pi} = 1$ if $\pi\in A$, and $b_{\pi} = 0$ if $\pi\not\in A$.
  Then together the two sums can be written as
  \[
  \sum\limits_{\pi} \Prob{(X,Y)\sim \mu}{ f(X,Y) = b_{\pi}, \Pi(X,Y) = \pi}\\
  \geq \sum\limits_{\pi} \left(\half - \eps\right)(1+\eps')^{-2}\left(\rho(\pi) - \frac{2^{4M}}{\eps'^2}\delta\eta(\pi)\right),
  \]
  where we used~\eqref{eq:7}. Since the sum of $\rho(\pi)$, as well as the sum of $\eta(\pi)$, is $1$, we
  get that the last expression is equal to $ \left(\half - \eps\right)(1+\eps')^{-2}(1-\frac{2^{4M}}{\eps'^2}\delta)\geq \half - \eps-\eps' -\frac{2^{4M}}{2\eps'^2}\delta$,
  and plugging this into~\eqref{eq:8} yields the result.
\end{proof}

\subsection{Proof of Theorem~\ref{thm:external}}
We can now combine Lemmas~\ref{lem:convert_low_universal} and~\ref{lem:step2} to deduce Theorem~\ref{thm:external}, as outlined below.

Suppose $(f,\mu)$ has $(\eps,\delta)$ relative-discrepancy with respect to $\rho$.
Let the error of $\Pi$ be denoted by $\eps'$; if $\eps'\geq \half - 2000\eps$ we are done, so assume otherwise.
Denote by $\eta = \half - \eps'$ the advantage of $\Pi$.
Using Lemma~\ref{lem:convert_low_universal} (choosing the $\eps$ there to be $\eta/160$), we get from $\Pi$ a protocol $\Pi'$ whose error
is at most $\eps'+\eta/4$ and has universal external information at most $M' = \frac{400M}{\eta}$. Thus,
the advantage of $\Pi'$ is at least $\half - (\eps' + \eta/4) = 3\eta/4$.
Using Lemma~\ref{lem:step2} (with $\eps'$ there to be $\eta/40$), we get that the advantage of $\Pi'$ is most
$20\eps+\frac{1}{2}\eta+\frac{20\cdot 40^2\cdot 2^{4M'}}{\eta^2}\delta$. Combining
the upper and lower bound on the advantage of $\Pi'$, we get
$3\eta/4\leq 20\eps+\frac{1}{2}\eta+\frac{20\cdot 40^2 \cdot 2^{4M'}}{\eta^2}\delta$.
Simplifying and using $\eta\geq 2000\eps$, we get that
$\frac{\eta}{10}\leq \frac{20\cdot 40^2 \cdot 2^{4M'}}{\eta^2}\delta$, and so
$\delta \geq  2^{-4M'}\eta^3/(10\cdot 20\cdot 40^2)\geq 2^{-5M'}$. Taking logarithm gives $5M'\geq \log(1/\delta)$ and
so by definition of $M'$ we get that $\eta\leq \frac{2000 M}{\log(1/\delta)}$.\qed

\section{Separating amortized zero-error communication complexity and external information}\label{sec:ammortized_separation}
In this section we prove Theorem~\ref{thm:main}, restated below.
\begin{reptheorem}{thm:main}
  For large enough $m$, there is
  a function $H\colon\power{N}\times\power{N}\to\power{}$ and a distribution $\mathcal{D}_H$ over inputs such that
  $\lim_{q\rightarrow \infty} \frac{1}{q} {\sf CC}_{\mathcal{D}_H^q}(H^q,0)\leq O(\sqrt{m}\log^2 m)$ and
  $\IC^{\sf external}_{\nu}[H,1/16]\geq \Omega\left(m\right)$.
\end{reptheorem}
We begin by presenting our construction, and then analyze it.
In particular, Lemmas~\ref{lem:small_ammortized} and~\ref{lem:large_external} imply the properties asserted by
Theorem~\ref{thm:main}.
\subsection{The construction}

\subsubsection{AND-OR trees}
Let $I_1,\ldots,I_{\sqrt{m}}$ be the partition of $[m]$ into equal sized sets given by $I_i = \sett{i\cdot\sqrt{m} + j}{j=1,\ldots, \sqrt{m}}$.
We define $h\colon\power{m}\to\power{}$ by
\[
h(z) = \bigwedge_{i=1}^{\sqrt{m}}\bigvee_{j\in I_i} z_j.
\]

We will need the following easy fact.
\begin{fact}\label{fact:certificate_h}
  For each $z\in\power{m}$, there exists a certificate for $h(z)$ of size $C(m)$, where $C(m) = O(\sqrt{m})$.
\end{fact}

Consider the function $h_{\land} \colon \power{m}\times \power{m}\to\power{}$ whose input is $(u,v)$, and
it is defined by $h_{\land}(u,v) = h(u_1\land v_1,\ldots,u_m\land v_m)$.
We will need the following result due to Jayram, Kumar and Sivakumar:
\begin{theorem}\label{thm:JKS}\cite{JKS}
  There exists a distribution $\mathcal{D}_h$ over $\power{m}\times \power{m}$, such that
  $\IC^{\sf internal}_{\mathcal{D}_h}[h_{\land},1/8]\geq m/100$.
\end{theorem}

\subsubsection{Our construction: AND-OR trees with a hint}
Fix a function $h$ as defined above, and let $\mathcal{D}_h$ be the distribution from Theorem~\ref{thm:JKS}.

Let $f_{\sf hint} \colon\power{n}\times\power{n}\to\power{}$ be a function, and $\mu = \half\mu_0 + \half\mu_1$ be a distribution over inputs such that
$(f_{\sf hint},\mu)$ has low relative discrepancy (we encourage the reader to think of the bursting noise function for sufficiently large $n$). Here, for each $b\in\power{}$,
$\mu_b$ is supported on $f_{\sf hint}^{-1}(b)$.
We define the function $H\colon\power{m}\times\power{m}\times(\power{n}\times\power{n})^{C(m)\log m}\to\power{}$ as
follows. View the input as $(u,x,v,y)\in\power{m}\times\power{n C(m)\log m}\times\power{m}\times \power{n C(m)\log m}$, and define
\[
H(u,x,v,y) = h_{\sf \land}(u,v)
\]

At first glance, one may wonder what is the role of the $x,y$-part of the input in the function $H$, as the definition of the function ignores them altogether.
The idea is that in the input distribution we consider, the inputs $x,y$ will be used as a pointer to a certificate of $h$ on the input
$z = (u_1\land v_1,\ldots,u_m\land v_m)$, and we will be able to compute this pointer with low amortized communication complexity
(but with error). Once the players compute the certificate set $I[z]$ of $h$ on $z$, they are able to communicate $C(m)\log m$ bits in
order to learn $z_i$ for each $i\in I[z]$. Thus, if the players were successful in computing the certificate set $I[z]$,
they above protocol would compute $H(u,x,v,y)$ with no error.

\paragraph{The input distribution $\mathcal{D}_H$.}
A sample according to the distribution $\mathcal{D}_H$ is drawn in the following way. First, sample $(u,v)\sim \mathcal{D}_{h}$, and consider the point
$z = (u_1\land v_1,\ldots,u_m\land v_m)$. By Fact~\ref{fact:certificate_h}, there is a certificate for $h(z)$ of size at most $C(m)$. Choose one such
certificate (in some canonical way), and let $\alpha(u,v) = (\alpha_1,\ldots,\alpha_{C(m)\log m})$ be a binary encoding of it.
For each $i=1,\ldots,C(m)\log m$ independently, we choose $(x^{i},y^{i})\sim \mu_{\alpha_i}$.
The sample of $\mathcal{D}_H$ is now $(u,x,v,y)$.

\paragraph{Choice of the parameters.}
Let $k\in\mathbb{N}$ be a large parameter, and choose $f_{\sf hint}$ to be the bursting noise function on the tree of height $2^{4^k}$ and $n$ accordingly.
Finally, we choose $m = 2^{k/10}$.

\subsection{The zero-error protocol}
In this section, we show a zero-error protocol of low amortized communication complexity as asserted in Theorem~\ref{thm:main}.
To show that, we will need the following result from~\cite{BravermanRao}.

\begin{thm}\label{thm:BR}[\cite{BravermanRao}]
  $\lim_{q\rightarrow \infty} \CC_{\mu^q}[f^q,\eps]/q = \IC^{\sf internal}_\mu[f,\eps]$.
\end{thm}
We also need the following easy fact.
\begin{fact}\label{fact:dist_transfer}
  Suppose $\mu,\mu_0,\mu_1$ are distributions such that $\mu = \half\mu_0 + \half\mu_1$,
  and let $\nu$ be any convex combination of $\mu_0,\mu_1$. Then
  $\IC^{\sf internal}_\nu[f,2\eps]\leq 2\IC^{\sf internal}_\mu[f,\eps]+6$.
\end{fact}
\begin{proof}
  Let $\Pi$ be a protocol for $(f,\mu)$ with $\eps$-error. Note that when executed on $\mu_0$ (or on $\mu_1$), the
  protocol has at most $2\eps$ error, and hence its error on $\nu$ is also at most $2\eps$. We next upper bound the internal
  information cost of $\Pi$ when executed on $\nu$.

  Consider the random variables $D, X, Y$, where we first sample $D\in \power{}$
  uniformly, then sample $(X,Y)\sim \mu_D$. Then
  $\MI[\Pi; X | Y] \geq \MI[\Pi; X | Y, D] - 1$ and similarly $\MI[\Pi; Y | X] \geq \MI[\Pi; Y | X, D] - 1$, so we get that
  \[
  \MI^{\sf internal}_\mu[\Pi] \geq \half \MI^{\sf internal}_{\mu_0}[\Pi]+\half \MI^{\sf internal}_{\mu_1}[\Pi] - 2,
  \]
  hence $\max(\MI^{\sf internal}_{\mu_0}[\Pi],\MI^{\sf internal}_{\mu_1}[\Pi])\leq 2\MI^{\sf internal}_\mu[\Pi]+4$.

  Since $\nu$ is a convex combination of $\mu_0$ and $\mu_1$, we may write $\nu = \lambda \mu_0 + (1-\lambda)\mu_1$
  and define random variables $D',X',Y'$ where: $D'=0$ with probability $\lambda$ and otherwise $D'=1$, and then we
  sample $(X',Y')\sim\mu_D$. We thus have
  \[
  \MI[\Pi(X',Y'); X' | Y'] \leq \MI[\Pi(X',Y'); X' | Y', D'] +1,~~
  \MI[\Pi(X',Y'); Y' | X'] \leq \MI[\Pi(X',Y'); Y' | X', D'] +1,
  \]
  and so
  \[
  \MI^{\sf internal}_\nu[\Pi]
  \leq \lambda \MI^{\sf internal}_{\mu_0}[\Pi] + (1-\lambda)\MI^{\sf internal}_{\mu_1}[\Pi]+2
  \leq \max(\MI^{\sf internal}_{\mu_0}[\Pi],\MI^{\sf internal}_{\mu_1}[\Pi])+2
  \leq 2\MI^{\sf internal}_\mu[\Pi]+6.\qedhere
  \]
\end{proof}

\begin{lemma}\label{lem:small_ammortized}
  $\lim_{q\rightarrow \infty} \frac{1}{q} \CC_{\mathcal{D}_H^q}(H^q,0) \leq O(\sqrt{m}\log^2 m)$
\end{lemma}
\begin{proof}
  Let $q\in\mathbb{N}$ be the number of copies of $H$ we wish to compute (thought of as very large), and denote the
  $q$-inputs to $H$ by
  $(u(1),x(1),v(1),y(1)),\ldots,(u(q),x(q),v(q),y(q))$. Set $r = C(m) \log m$; we now define $r$
  distinct $q$-tuples of inputs for $f$. For each $i = 1,\ldots, r$, consider the $q$-tuple of inputs for
  $f$ resulting from taking the $i$th input from each one of $(x(1),y(1)),\ldots,(x(q),y(q))$, i.e.
  $a(i) = (x(1)^i,\ldots,x(q)^i)$ and $b(i) = (y(1)^i,\ldots,y(q)^i)$. We note that for each $i$,
  the distribution of $(a(i),b(i))$ is a product distribution $\nu_i^q$ where $\nu_i$ is some convex
  combination of $\mu_0,\mu_1$.

  Combining Lemma~\ref{lemma:GKR_internal} and Fact~\ref{fact:dist_transfer} we get that for each
  $i$, $\text{IC}^{\sf internal}_{\nu_i}[f,2^{1-k}] = O(k)$, and so by Theorem~\ref{thm:BR}, there is $q(i)$ such that
  for all $q\geq q(i)$, we may find a protocol $\Pi_i$ communicating $O(q k)$ bits in expectation whose error
  on each copy of $f^{q}$ on $\nu_i^q$ is at most $2^{1-k}$. For the rest of the proof, we take
  $q\geq \max_{i=1,\ldots,r} q(i)$.

  We are now ready to present the protocol for $H$. For each $i$, the players use the protocol
  $\Pi_i$ in order to compute $f^q(a(i),b(i))$. Thus, the players now have a candidate answer for
  $f(x(j)^i,y(j)^i)$ for each $j=1,\ldots,q$ and $i=1,\ldots,r$.
  The players now check for each $j=1,\ldots,q$ the subset of coordinates that $(f(x(j)^i,y(j)^i))_{i=1,\ldots,r}$
  encodes, call it $A_i$, and then communicate all of the bits in $u(j),v(j)$ that are in it, i.e.
  $u(j)^\ell, v(j)^\ell$ for $\ell\in A_i$. If this partial assignment to $h_{\land}$ is indeed a
  certificate -- the players declare the copy $i$ to be successful and thus know the value
  $H(u(j),x(j),v(j),y(j))$. Otherwise, if copy $j$ is not successful, the players use the trivial protocol
  for $H$ and exchange $u(j),v(j)$ fully to compute $H$ on that copy (that simply exchanges the players' inputs).

  \paragraph{Correctness.} It is easy to see that the players are always correct on each copy. They are correct on a ``successful copy''
  by the definition of certificates, and are correct on an ``unsuccessful copy'' since they exchange all of the relevant input bits needed to
  compute $h_{\land}$ on each such copy.

  \paragraph{Average communication complexity.} Let $E_i$ be the event that copy $i$ is successful, and
  let $Q = \sum\limits_{i=1}^{q} 1_{E_i}$ be the number of successful copies. Also, denote by $A$ the total number of
  bits communicated by the simulation of $\Pi_i$ for $i=1,\ldots,r$. With these notations, we note that the random variable
  $A + (q-Q) 2m$ bounds the total number of bits communicated by the protocol.

  First, note that by choice of $\Pi_i$ we have that $\Expect{}{A}\leq O(q k r)$. Secondly we lower bound
  the expectation of $Q$.
  Note that for each $i,j$, the probability the players computed $f(x(j)^i,y(j)^i)$ correctly is at least
  $1-2^{1-k}$. Thus, for each $j$, the probability that they computed $f(x(j)^i,y(i)^i)$ correctly for
  all $i=1,\ldots,r$ is at least $1-r2^{1-k}\geq 1-2^{-9k/10}$. Thus, $\Expect{}{Q}\geq (1-2^{-9k/10})q$.

  Overall, we get that the expected number of bits communicated is at most
  \[
  \Expect{}{A + (q-Q) 2m}
  \leq O(q k r) + 2mq 2^{-9k/10}
  = O(qk r),
  \]
  where in the last inequality we used the fact that $m=2^{k/10}$. Therefore, we get that
  \[
  \lim_{q\rightarrow \infty} \frac{1}{q} \CC_{\mathcal{D}_H^q}(H^q,0)
  =O(kr)
  =O(\sqrt{m} \log^2 m).\qedhere
  \]

\end{proof}

\subsection{Lower bounding the external information of $H$}
In this section, we prove the lower bound on the external information of $(H,\mathcal{D}_H)$ as asserted in Theorem~\ref{thm:main}.
The main step in this proof is Lemma~\ref{lem:drop_hint}, which asserts that any protocol for $(H,\mathcal{D}_H)$ with low external
information can be converted into a protocol for $(h_{\sf final},\mathcal{D}_h)$ with low external information. The second step
of the proof is to argue that the latter is impossible, and is essentially the content of Theorem~\ref{thm:JKS} (up to the choice
of the parameters).

\subsubsection{Dropping the hints}
\begin{lemma}\label{lem:drop_hint}
  Let $M\in\mathbb{N}$, $\eps,\eps',\xi,\delta>0$, let $h_{\sf \land} \colon\power{m}\times\power{m}\to\power{}$ be as above
  and let $f \colon\power{n}\times\power{n}\to\power{}$ be a function with $(\eps,\delta)$ relative discrepancy.

  Then any protocol $\Pi$ for $(H,\mathcal{D}_H)$ with error $\eps'$ and $\MI_{\mathcal{D}_H}^{{\sf external}}[\Pi]\leq M$,
  can be converted into a protocol $\Pi'$ for $(h_{\sf \land},\mathcal{D}_h)$ satisfying:
  \begin{enumerate}
    \item The protocol $\Pi'$ has universal external information at most $\frac{4M}{\xi}$.
    \item The error of $\Pi'$ is at most $\eps' + 61\xi + 20C(m)\log m \cdot \eps + 2^{33 C(m)\log m M/\xi^2} \cdot \delta$.
  \end{enumerate}
\end{lemma}
The rest of this section is devoted to the proof of Lemma~\ref{lem:drop_hint}.
Denote $r = C(m)\log m$.

Recall the input
distribution $\mathcal{D}_H$ of the function $H$, and consider the distribution $\tilde{\mathcal{D}}_H$ defined as follows.
To draw a sample, we take $(u,v)\sim \mathcal{D}_{h}$, and for each $i=1,\ldots,r$ take $(x^i,y^i)\sim \mu$
independently, and output $(u,x,v,y)$.

By Lemma~\ref{lem:convert_low_universal}, we may convert $\Pi$ into a protocol
$\Gamma$ that has error at most $\eps'' = \eps' + 40\xi$ on $(H,\mathcal{D}_H)$
and has universal external information at most $M' = 4M/\xi$. We would like to show that the
distribution over the transcript $\Gamma(U,X,V,Y)$ when $(U,X,V,Y)\sim \mathcal{D}_H$,
is statistically close to the distribution of the transcript $\Gamma(U,\tilde{X},V,\tilde{Y})$
where $(U,\tilde{X},V,\tilde{Y})\sim \tilde{\mathcal{D}}_H$.

To do so, we consider the hybrid ensembles of random variables. That is, sample $(U,X,V,Y)\sim \mathcal{D}_H$
as well as $\tilde{X},\tilde{Y}\sim \mu^{C(m)\log m}$, so that the distribution of $(U,\tilde{X},V,\tilde{Y})$
is $\tilde{\mathcal{D}}_H$. Denote
\[
\mathcal{X}^i = (\tilde{X}_1,\ldots,\tilde{X}_{i}, X_{i+1}, X_{i+2},\ldots, X_r), \qquad
\mathcal{Y}^i = (\tilde{Y}_1,\ldots,\tilde{Y}_{i}, Y_{i+1}, Y_{i+2},\ldots, Y_r)
\]
and define $\Gamma_i = \Gamma(U,\mathcal{X}^i,V,\mathcal{Y}^i)$. We show that
the statistical distance between $\Gamma_i$ and $\Gamma_{i+1}$ is small.

\begin{claim}\label{claim:hybrid}
  For all $i \in\set{0,\ldots,r-1}$, we have that
  ${\sf SD}(\Gamma_i, \Gamma_{i+1})\leq 20\eps + 21\frac{\xi}{r} + 20\frac{r^2 2^{8M'r/\xi}}{\xi^2}\delta$.
\end{claim}
\begin{proof}
  Since $\Gamma$ has universal external information at most $M'$, it follows from Lemma~\ref{lem:low_universal}
  that the external information of $\Gamma$ when run on $(U,\mathcal{X}^{i+1},V,\mathcal{Y}^{i+1})$ is at most $2M'$, i.e.
  \[
  \MI[\Gamma(U,\mathcal{X}^{i+1}, V,\mathcal{Y}^{i+1}); \mathcal{X}^{i+1}, \mathcal{Y}^{i+1}, U,V]\leq 2M'.
  \]
  For a tuple $Z = (Z_1,\ldots,Z_r)$ and $i\in[r]$, we denote by $Z_{\neq i}$ the tuple obtained from $Z$ by dropping
  the $i$th coordinate. By Fact~\ref{fact:MI_cond}, the left hand side is at least
  \[
  \MI[\Gamma(U,\mathcal{X}^{i+1},V,\mathcal{Y}^{i+1});
  \tilde{X}_{i+1}, \tilde{Y}_{i+1}
  ~|~ \mathcal{X}^{i+1}_{\neq i+1}, Y^{i+1}_{\neq i+1}, U, V].
  \]
  For each $x_{\neq i+1}, y_{\neq i+1}, u,v$ we define
  \[
  M[x_{\neq i}, y_{\neq i}, u,v]
  =\MI[\Gamma(U,\mathcal{X}^{i+1},V, \mathcal{Y}^{i+1});
  \tilde{X}_{i+1}, \tilde{Y}_{i+1}
  ~|~ \mathcal{X}^{i+1}_{\neq i+1}=x_{\neq i+1}, \tilde{Y}^{i+1}_{\neq i+1}=y_{\neq i+1}, U=u, V=v].
  \]
  Then we have that $\Expect{}{M[\mathcal{X}^{i+1}_{\neq i+1}, \tilde{Y}^{i+1}_{\neq i+1}, U,V]}\leq 2M'$, and thus
  by Markov's inequality we have that $M[\mathcal{X}^{i+1}_{\neq i+1}, \tilde{Y}^{i+1}_{\neq i+1}, U,V]\leq 2r M'/\xi = M''$
  with probability at least $1-\xi/r$; denote this event by $E$.

  \skipi
  Take $(u,v,x_{\neq i+1},y_{\neq i+1})\in E$ and condition on
  $\mathcal{X}^{i+1}_{\neq i+1} = x_{\neq i+1}$ , $\mathcal{Y}^{i+1}_{\neq i+1} = y_{\neq i+1}$, $U=u$ and $V=v$.
  Note that the distribution of $(\tilde{X}_{i+1}, \tilde{Y}_{i+1})$ is $\mu$, and that
  $\Gamma(U,\mathcal{X}^{i+1}, V,\mathcal{Y}^{i+1})$ is a protocol
  whose input is $\tilde{X}_{i+1}, \tilde{Y}_{i+1}$ and has external information
  $M[x_{\neq i}, y_{\neq i}, u,v]$ with its the inputs. Therefore,
  by Lemma~\ref{lem:step2} (applied with $\eps' = \xi/r$) we get that
  \[
  {\sf SD}\left(
  \Gamma(U,\mathcal{X}^{i+1},V, \mathcal{Y}^{i+1})|_{f(\tilde{X}_{i+1}, \tilde{Y}_{i+1}) = 0},
  \Gamma(U,\mathcal{X}^{i+1},V, \mathcal{Y}^{i+1})|_{f(\tilde{X}_{i+1}, \tilde{Y}_{i+1}) = 1}\right)
  \leq 20\left(\eps + \frac{\xi}{r} + \frac{r^2 2^{4M''}}{\xi^2}\delta\right).
  \]
  Let $\alpha_1,\ldots,\alpha_r$ be the encoding of the certificate chosen for $h(u,v)$, and let
  $b = \alpha_{i+1}$. Then it follows that
  \[
  {\sf SD}\left(
  \Gamma(U,\mathcal{X}^{i+1},V, \mathcal{Y}^{i+1})|_{f(\tilde{X}_{i+1}, \tilde{Y}_{i+1}) = b},
  \Gamma(U,\mathcal{X}^{i+1},V, \mathcal{Y}^{i+1})\right)
  \leq 20\left(\eps + \frac{\xi}{r} + \frac{r^2 2^{4M''}}{\xi^2}\delta\right).
  \]
  Taking average over $\mathcal{X}^{i+1}_{\neq i+1}, \mathcal{Y}^{i+1}_{\neq i+1}, U,V$
  and noting that the event $E$ fails with probability at most $\xi/r$ (and then the statistical distance is at most $1$),
  we get that
  \[
  {\sf SD}\left(
  \Gamma(U,\mathcal{X}^{i+1},V, \mathcal{Y}^{i+1})|_{f(\tilde{X}_{i+1}, \tilde{Y}_{i+1}) = \alpha_{i+1}(U,V)},
  \Gamma(U,\mathcal{X}^{i+1},V, \mathcal{Y}^{i+1})\right)
  \leq 20\eps + 21\frac{\xi}{r} + 20\frac{r^2 2^{4M''}}{\xi^2}\delta.
  \]
  The statement of the claim now follows since the distribution $\Gamma(U,\mathcal{X}^{i+1},V, \mathcal{Y}^{i+1})|_{f(\tilde{X}_{i+1}, \tilde{Y}_{i+1}) = \alpha_{i+1}(U,V)}$
  is precisely $\Gamma_i$, and the distribution of $\Gamma(U,\mathcal{X}^{i+1},V, \mathcal{Y}^{i+1})$ is precisely $\Gamma_{i+1}$.
\end{proof}

\begin{claim}\label{claim:hybrid_final}
  ${\sf SD}(\Gamma(U,X,V,Y), \Gamma(U,\tilde{X},V,\tilde{Y}))\leq 20r\eps + 21\xi + 20\frac{r^3 2^{8M'r/\xi}}{\xi^2}\delta$.
\end{claim}
\begin{proof}
  Since $\Gamma_0 = \Gamma(U,X,V,Y)$ and $\Gamma_r = \Gamma(U,\tilde{X},V,\tilde{Y})$, the statement of
  the claim follows by summing Claim~\ref{claim:hybrid} and using the triangle inequality.
\end{proof}

Set $\eta = 20r\eps + 21\xi + 20\frac{r^3 2^{8M'r/\xi}}{\xi^2}\delta$.
Using Claim~\ref{claim:hybrid_final}, since the error of $\Gamma_0$ in computing $(H,\mathcal{D}_H)$ is at most
$\eps''$, it follows that $\Gamma_r$ has error at most $\eps'' + \eta$ in computing
$(H,\tilde{\mathcal{D}}_H)$. Note that the distribution of $\tilde{X},\tilde{Y}$ is
completely independent of $U,V$, and $(U,V)$ is distributed according to $\mathcal{D}_{h}$. Thus, we may find $\tilde{x},\tilde{y}$
such that the protocol $\Pi'(U,V) \defeq \Gamma(U,\tilde{x},V,\tilde{y})$ has error at most $\eps'' + \eta$ in computing $(h_{\sf \land}, \mathcal{D}_{h})$.
Bounding $\eta \leq 20r\eps + 21\xi + 2^{33r M/\xi^2} \cdot \delta$ and plugging in $r$ gives the claimed bound.\qed

\subsubsection{Concluding the external information lower bound}
\begin{lemma}\label{lem:large_external}
  $\IC^{\sf external}_{\mathcal{D}_H}[H,1/16]\geq \Omega(m)$.
\end{lemma}
\begin{proof}
  Let $\xi = 1/1000$, and suppose we have a protocol $\Pi$ for $(H,\mathcal{D}_H)$ with external information at most $\xi m/800$ and error $\eps'$.
  Using Lemma~\ref{lem:drop_hint}, we find a protocol $\Pi'$ for $(h_{\land}, \mathcal{D}_h)$ with external information at most $m/200$, whose error is upper bounded
  as in Lemma~\ref{lem:drop_hint}. However, by Fact~\ref{fact:external_greater_internal} and Theorem~\ref{thm:JKS} we have that
  $\IC^{\sf external}_{\mathcal{D}_h}[h_{\land},1/8]\geq \IC^{\sf internal}_{\mathcal{D}_h}[h_{\land},1/8]\geq m/100$, so $\Pi'$ must have error
  at least $1/8$. Combining the upper and lower bounds on the error of $\Pi$ yields that
  \[
  \frac{1}{8}\leq \eps' + 61\xi + 20C(m)\log m \cdot \eps + 2^{33 m\cdot C(m)\log m /\xi^2} \cdot \delta
  \]
  By the choice of parameters, we have that $m=2^{k/10}$ and that $f$ has $(\eps,\delta)$ relative-discrepancy
  for $\eps = 2^{-k}$, $\delta = \eps/2^{2^k}$ (by Lemma~\ref{lem:GKR}), so
  $20C(m)\log m \cdot \eps + 2^{33 C(m)\log m M/\xi^2} \cdot \delta \leq 2^{-k/2} + 2^{m^2}\cdot \delta\leq 2^{-k/4}$.
  Thus we get that $\eps ' \geq \frac{1}{8} - 61\xi - 2^{-k/4}\geq \frac{1}{16}$, as desired.
\end{proof}

\section{Tightness and additional implications}\label{sec:tight}
In this section, we prove Theorem~\ref{thm:quadratic_tight}, which asserts that Theorem~\ref{thm:main} is nearly-tight for protocols with
constant error. We then prove that for protocols with zero-error, a better, arbitrarily large, separation holds in the form of
Corollary~\ref{corr:zero_error_external_separation}.
\subsection{Proof of Theorem~\ref{thm:quadratic_tight}}
 Fix $(f,\mu)$ as in the Theorem, denote $A = \lim_{q\rightarrow \infty} \frac{1}{q} \CC_{\mu^q}(f^q,0)$
 and take $q$ such that $\CC_{\mu^q}(f^q,0)\leq 2A q$. Thus, there is a zero-error protocol $\Pi$
 for $f^q$, whose average communication cost on $\mu^q$ is at most $2Aq$. Let $(X_1,Y_1),\ldots,(X_q,Y_q)\sim\mu$
 be independent. Then the independence implies that $\MI[\Pi; X_i,Y_i]\leq \MI[\Pi; X_i,Y_i~|~X_{<i}, Y_{<i}]$,
 and so by Fact~\ref{fact:MI_chainrule}
 \[
 \Expect{i\in[q]} {\MI[\Pi; X_i,Y_i]}
 \leq \frac{1}{q}\sum\limits_{i=1}^{q}{\MI[\Pi; X_i,Y_i~|~X_{<i}, Y_{<i}]}
 = \frac{1}{q}\MI[\Pi; X, Y]
 \leq
 \frac{1}{q}\HH[\Pi]
 \leq \frac{1}{q}\Expect{}{\card{\Pi}}\leq 2A.
 \]
 It follows that there is an $i\in[q]$ such that
 $\MI[\Pi; X_i,Y_i]\leq 2A$, and without loss of generality assume $i=1$ is such copy. We first handle
 the simple case in which $A\leq \eps\log(1/\eps)/10$. In this case, by the Data Processing inequality we have
 $\MI[\Pi; f(X_1,Y_1)]\leq 2A$, and since $\Pi$ has zero-error we get that $\HH[f(X_1,Y_1)]\leq 2A$.
 Thus, $f(X_1,Y_1)$ is close to constant, i.e. there is $b\in\power{}$ such that $\Prob{(X_1,Y_1)\sim\mu}{f(X_1,Y_1) = b}\geq 1-\eps$,
 and we have a trivial protocol for $f$ (in particular $\MI^{\sf external}_{\mu}[f,\eps] = 0$). Therefore, we may assume for the
 rest of the proof that $A\geq \eps\log(1/\eps)/10$.

 Define the set of good tuples $G = \sett{(x_1,y_1)}{\DKL{\Pi_{X_1 = x_1, Y_1 = y_1}}{\Pi}\leq 2A/\eps}$, and note that
 using Fact~\ref{fact:mutual_div_KL} and Markov's inequality yields that $\mu(G)\geq 1-\eps$.

 Denote $X = (X_1,\ldots,X_q)$, $Y = (Y_1,\ldots,Y_q)$.
 For a transcript $\pi = (a,b)$ where $a = (a_1,\ldots,a_m)$ are the messages of Alice and $b = (b_1,\ldots,b_m)$ are the
 messages of Bob, and inputs $x,y$ define
 \[
  P_{A}^x(a,b) = \prod\limits_{j< m} \cProb{(X,Y)\sim \mu}{A_{<j} = a_{<j}, B_{<j} = b_{<j}, X = x}{A_j = a_j},
 \]
 and similarly
 \[
  P_{B}^y(a,b) = \prod\limits_{j< m} \cProb{(X,Y)\sim \mu}{A_{\leq j} = a_{\leq j}, B_{<j} = b_{<j}, Y = y}{B_j = b_j}.
 \]

 We will use the protocol $\Pi$ to construct a protocol $\Pi'$ for $(f,\mu)$, whose input is $(x_1,y_1)$, but first let us introduce
 some terminology and make some observations. We say a transcript $\pi = (a,b)$ of $\Pi$ is compatible with
 Alice, if there exists $(x_2,\ldots,x_q)$ such that $P_{A}^x(\pi)>0$ for $x = (x_1,\ldots,x_n)$, and analogously for Bob. Note that
 if $\pi$ is compatible with both Alice and Bob, then there are extensions $x$ of $x_1$, and $y$ of $y_1$, such that the probability that
 $\Pi(x,y) = \pi$ is $P_{A}^x(\pi)P_{B}^y(\pi)>0$. Since $\Pi$ has
 zero-error, it means that in that case the value of $f(x_1,y_1)$ is computed correctly in transcript $\pi$.

 For each $\pi$, let $R_{\pi}$ be the set of $(x_1,y_1)$ for which $\pi$ is compatible with both Alice and Bob,
 and note that $R_{\pi}$ is a monochromatic rectangle of $f$.
 We say a transcript $\pi$ is a $0$-transcript if the value of $f(x_1,y_1)$ on $R_{\pi}$ is $0$, and say it is a $1$-transcript
 if the value of $f(x_1,y_1)$ on $R_{\pi}$ is $1$. Note that if $\pi_0$ is a $0$-transcript, and $\pi_1$ is a $1$-transcript,
 then $R_{\pi_0}$ and $R_{\pi_1}$ are disjoint.

 We need the following claim, asserting that a randomly chosen transcript is compatible with both players with noticeable probability.
 \begin{claim}\label{claim:cover}
 Let $(x_1,y_1)\in G$. Then
   \[
   \Prob{\pi\sim \Pi(X_1,Y_1)}{\pi\text{ is compatible with both Alice and Bob on inputs $(x_1,y_1)$}}\geq  2^{-2A/\eps - 17}.
   \]
 \end{claim}
 \begin{proof}
  Define $p(\pi) = \Prob{(X_1,Y_1)\sim \mu}{\Pi(X_1,Y_1) = \pi}$ and
  $p(\pi~|~x_1,y_1) = \Prob{}{\Pi(x_1,y_1) = \pi}$, and consider
  \[
  H = \sett{\pi\in{\sf Supp}(\Pi)}{ 0<p(\pi~|~x_1,y_1)\leq 2^{2(A/\eps+8)} p(\pi)}.
  \]
  By Fact~\ref{fact:KL_abs_val} and the definitions of $G$ and $H$,
  \begin{align*}
  \sum\limits_{\pi\not\in H}p(\pi~|~x_1,y_1)\cdot 2\left(\frac{A}{\eps}+8\right)
  \leq
  \sum\limits_{\pi}{p(\pi~|~x_1,y_1)\card{\log\left(\frac{p(\pi~|~x_1,y_1)}{p(\pi)}\right)}}
  &\leq \DKL{\Pi_{X_1 = x_1, Y_1 = y_1}}{\Pi} + 8\\
  &\leq\frac{A}{\eps} + 8,
  \end{align*}
  so
  $\sum\limits_{\pi\in H} p(\pi~|~x_1,y_1)\geq 1/2$. It follows that
  \begin{align*}
   \Prob{\pi\sim \Pi(X_1,Y_1)}{\pi\text{ is compatible with both Alice and Bob on inputs $(x_1,y_1)$}}
   &\geq \sum\limits_{\pi\in H} p(\pi)\\
   &\geq \sum\limits_{\pi\in H} 2^{-2(A/\eps+8)}p(\pi~|~x_1,y_1)\\
   &\geq 2^{-2A/\eps - 17}.\qedhere
  \end{align*}
 \end{proof}

 The idea of $\Pi'$ is to consider a long enough list of possible transcripts of $\Pi$, such that almost all input tuples $(x,y)\in G$
 have a compatible transcript in the list. Thus, we will have a collection of $L$ monochromatic rectangles that cover most of the mass
 of $\mu$, and we may invoke the classical argument from~\cite{AUY} that constructs a protocol from monochromatic rectangles. We outline
 the argument below for completeness.

 Using Claim~\ref{claim:cover}, there is a list of $L = 2^{2A/\eps + 17} \log(1/\eps) = 2^{O(A/\eps)}$
 possible transcripts of $\Pi$, $\pi_1,\ldots,\pi_L$, and $G' \subseteq G$ with $\mu(G')\geq \mu(G)-\eps\geq 1-2\eps$,
 such that for each $(x_1,y_1)\in G'$ there is $i\in[L]$ such that $\pi_i$ is compatible with both Alice and Bob on $(x_1,y_1)$.
 We now describe $\Pi'(X_1,Y_1)$. The players will try to convince themselves that $f(x_1,y_1) = 0$,
 and for that they will try to eliminate from the list all $1$-transcripts. Formally, at each step of the protocol there is an active set of $1$-transcripts,
 $S\subseteq[L]$, and the goal of the players at each step is either to shrink the size of $S$ by factor $2$, or learn the value of
 $f(x_1,y_1)$.

 Suppose that both players are compatible with a $0$-transcript $\tau_0$ from the list, and consider all $1$-transcripts $\tau_1$ in $S$. Note
 that since $R_{\tau_0}$ and $R_{\tau_1}$ are disjoint, either their $x$-range is disjoint, or their $y$-range is disjoint. In particular, it follows
 that either for the $x$-range or $y$-range -- say $x$-range, at least half of the $R_{\tau_1}$'s are disjoint from
 $R_{\tau_0}$ in it. In this case, we say that $\tau_0$ eliminates half of the $\tau_1$'s from the point of view of Alice.

 Thus, at each step, the player considers all $0$-transcripts from the list that are compatible with their input, and checks whether there is
 at least one that eliminates half of the $1$-transcripts in $S$ from their view.
 \begin{enumerate}
   \item If there is, the player chooses one such $\tau_0$ arbitrarily, and sends
    the index of that transcript in the list. All of the $1$-transcripts that are inconsistent with $\tau_0$ from that player's point of view
    are discarded from $S$. If $\tau_0$ is also compatible with the other player, the players declare the output to be $0$ and terminate, and otherwise
    the players continue in the protocol.
   \item  If there is no such $\tau_0$, the player indicates so and passes the turn to the other player.
 \end{enumerate}
 If both players passed as in item 2 above in consecutive turns, the players declare the output of the function to be $1$. Otherwise,
 the protocol continues until $S$ becomes empty, in which case the players declare the output to be $0$.

 \paragraph{Analyzing the communication complexity.}
 Note that in each turn, the number of bits communicated is at most $O(A/\eps)$, so to bound the communication complexity of $\Pi'$ we need
 to bound the number of turns. Since the list $S$ starts off at size $2^{O(A/\eps)}$ and shrinks by factor $2$ at least once every $3$ turns,
 it follows that the number of rounds is $O(A/\eps)$, and so the communication complexity of $\Pi'$ is at most $O\left(\frac{A^2}{\eps^2}\right)$.
 In particular, we get that $\MI^{\sf external}_{\mu}[\Pi']\leq O\left(\frac{A^2}{\eps^2}\right)$.

 \paragraph{Correctness.} We claim the $\Pi'(x_1,y_1)$ is correct whenever $(x_1,y_1)\in G'$, and since $\mu(G')\geq 1-2\eps$ it follows that
 it has error at most $2\eps$ when ran on $\mu$. Indeed, note that if $(x_1,y_1)\in G'$ and $f(x_1,y_1) = 0$, then there is $\tau_0$ in the
 list that is compatible with both of them, hence the players will never pass in two consecutive turns and the output of the protocol will be
 $0$. If $f(x_1,y_1) = 1$, then there is $\tau_1$ a $1$-transcript that is compatible with both players, hence it will never be removed from $S$.
 Thus, the output of the protocol will be $1$.

\subsection{Proof of Corollary~\ref{corr:zero_error_external_separation}}
 Fix $k$, and pick $(f,\mu)$ from Theorem~\ref{thm:main}. Define $f'\colon\power{n+1}\times\power{n+1}\to\power{}$ in the following way:
 view the input $(x',y')$ as $x' = (x,a)$, $y' = (y,b)$ where $(x,y)\in \power{n}\times\power{n}$, and define
 $f'(x,y) = f(x,y)$ if $a = b = 1$ and $0$ otherwise.
 Set $p = \frac{1}{\sqrt{k}\log^2 k}$ and consider the distribution $\nu$ that puts $(1-p)$ mass uniformly on
 $x' = (x,a), y'=(y,b)$ such that $a = b = 0$, and puts the rest of its mass on $x',y'$ where $a = b = 1$, i.e.
 $\nu(x',y') = p\mu(x,y)$ for any such $x',y'$. We claim that $(f,\nu)$ satisfies the properties asserted by Corollary~\ref{corr:zero_error_external_separation}.

  Let $\Pi$ be a zero-error protocol for $(f,\nu)$, and consider the random variables $(D,X' = (X,A),Y' = (Y,B))$ where $D = 0$ with probability $p$ and otherwise $D=1$.
  If $D = 0$ we sample $(X,Y)\sim\mu$ and set $A = B = 1$, and otherwise we sample $X$ and $Y$ uniformly form $\power{n}$ and set $A = B = 0$. Then
  \[
  \MI^{\sf external}_{\nu}[\Pi]
  =\MI[\Pi; X,Y]
  \geq
  \MI[\Pi; X',Y'~|~D] - 1
  \geq
  p \MI[\Pi; X',Y'~|~D = 0] - 1,
  \]
  so $\IC^{\sf external}_{\nu}[f',0]\geq p\IC^{\sf external}_{\mu}[f,1/16] - 1\geq \Omega(\sqrt{k}/\log^2 k)$.

  As for a zero-error protocol
  with $O(1)$ amortized communication on $\nu$, by assumption there is a protocol $\Pi_q$ solving $(f^{m},\mu^m)$ with
  zero-error and expected communication complexity $O(m\sqrt{k}\log^2 k) + o(m)= O(m/p) + \alpha(m)$, where $\alpha(m)$ is monotone and
  $\alpha(m)/m\rightarrow 0$.

  Consider the following protocol $\Pi$
  for $(f^q, \nu^q)$. Denote the inputs by $(x'^1,y'^1),\ldots,(x'^q,y'^q)$, and write $x'^i = (x^i, a^i)$ and $y'^i = (y^i, b^i)$.
  Alice first identifies all copies $i$ such that $a^{i} = 0$,
  and communicates them to Bob, so that the answer to that copy is $0$. Similarly, Bob identifies
  all copies $i$ where $b^{i} = 0$, and communicates them to Alice. Let $A$ be the set of $i$'s such that
  $a^{i} = b^i = 1$ (note that both Alice and Bob know $A$), and set $m = \card{A}$. The players use the protocol $\Pi_m$ to solve the copies of $f$
  corresponding to $A$.

  \paragraph{Correctness.} It is clear that the protocol $\Pi$ is correct and has zero error.
  \paragraph{Amortized communication cost.} Let $Q = \card{A}$ be a random variable. Note that the number of bits transmitted in the phase of the protocol
  in which trivial copies are identified, is $O(q)$ (corresponding to encodings of two subset of $[q]$).
  Thus, the expected communication cost of $\Pi$ is at most
  $O(q) + \Expect{}{\CC_{\mu'^Q}(\Pi_{Q})}$. Conditioning on $Q$, we have that $\CC_{\mu'^Q}(\Pi_{Q}) =  O(\card{Q}/p) + \alpha(\card{Q})\leq O(\card{Q}/p)+\alpha(q)$.
  Note that
  $\Expect{}{\card{Q}} = p q$, so we get that
  \[
  \Expect{}{\CC_{\mu'^Q}(\Pi_{Q})}
  \leq \Expect{}{O(\card{Q}/p) + \alpha(q)}
  \leq O(q) + \alpha(q).
  \]
  In conclusion, the expected communication cost of $\Pi$ is $O(q) + \alpha(q)$, and since $\alpha(q)/q\rightarrow 0$ we conclude that the amortized communication cost of $\Pi$ is
  $O(1)$.

\bibliographystyle{alpha}
\bibliography{ref}
\appendix
\section{Missing proofs}

\subsection{Upper bounding amortized zero-error communication complexity}
\label{sec:missing}

\begin{thm}\label{thm:upper_bd_easy}
  For every total function $f\colon\power{n}\times\power{n}\to\power{}$, it holds that
  \[
  \lim_{q\rightarrow \infty} \frac{1}{q} \CC_{\mu^q}(f^q,0)\leq \IC^{\sf external}_{\mu}[f,0].
  \]
\end{thm}
For the proof, we need the following lemma due to~\cite{HJMR}.
\begin{lemma}\label{lem:rejection_samp}
  Let $P,Q$ be distributions such that $\DKL{P}{Q} < \infty$. Then there exists a sampling procedure $\mathcal{P}$, that
  on an input $(q_1,q_2,\ldots)$ consisting of a list of independent samples from $Q$, outputs an index $r^{\star}$ such that
  the distribution of $q_{r^{\star}}$ is $P$, and the expected length of $r^{\star}$ is at most
  \[
  \DKL{P}{Q} + 2\log(\DKL{P}{Q} + 1) + O(1).
  \]
\end{lemma}

\begin{proof}[Proof of Theorem~\ref{thm:upper_bd_easy}]
  Let $\Pi$ be a zero-error protocol for $f$. We show that there is a sequence of zero-error protocols
  $(\Gamma_k)_{k\in\mathbb{N}}$ for $(f^k)_{k\in\mathbb{N}}$, such that $\CC[\Gamma_k] = k\MI^{\sf external}_{\mu}[\Pi] + o(k)$,
  from which the theorem clearly follows.

  Let $M = \CC_{\mu}(\Pi) < \infty$.
  Let $k$ be large, let $(X_1,Y_1),\ldots,(X_k,Y_k)\sim \mu$ be independently sampled (i.e. an input for $f^{k}$),
  and denote by $(x_1,y_1),\ldots,(x_k,y_k)\sim\mu$ a realization of them.
  The protocol $\Gamma_k$ we construct works by rounds, wherein in round $t$, the player that speaks in that round
  in $\Pi$ also speaks. The goal of that player is to communicate to the other player all of the bits that would be
  sent on that round in $\Pi(X_1,Y_1),\ldots,\Pi(X_k,Y_k)$ conditioned on the transcripts so far and the input of
  the speaking player.

  We describe $\Gamma_k$ more precisely now. For each $i\in [k]$ and step $t$, the players maintain a transcript of $\Pi(X_i,Y_i)$
  up to step $t$, which we denote by $\pi_{i,t}$. Thus, denoting $\pi_t = (\pi_{1,t},\ldots,\pi_{k,t})$, the
  goal of the players is to correctly sample $\pi_{t+1}$ conditioned on $\pi_{t}$, and to do so with low expected communication complexity low.
  Towards this end, we denote by $\Pi(X_i,Y_i)_{j}$
  the $j$th bit exchanged in the protocol $\Pi$ on $X_i,Y_i$, $X = (X_1,\ldots,X_k)$ and $Y = (Y_1,\ldots,Y_k)$.
  Denote
  \begin{align*}
  D(\pi_{t},t,x,y) =
  \mathrm{D}_{\text{KL}}\Big(\big(
  &\Pi(X_1,Y_1)_{t+1},\ldots,\Pi(X_k,Y_k)_{t+1}\big)|_{X=x, Y=y, \Pi(X,Y)_{\leq t} = \pi_t}\parallel\\
  &\big(\Pi(X_1,Y_1)_{t+1},\ldots,\Pi(X_1,Y_1)_{t+1}\big)|_{\Pi(X,Y)_{\leq t} = \pi_t}\Big).
  \end{align*}
  We now argue that the players can sample $\pi_{t+1}$ conditioned on $\pi_{i,t}$ by expectedly communicating at most
  $O(1) + D(\pi_{t},t) + 2\log(D(\pi_{t},t) + 1)$ bits. Note that regardless of what player speaks, the distribution
  $\big(\Pi(X_1,Y_1)_{t+1},\ldots,\Pi(X_k,Y_k)_{t+1}\big)|_{\Pi(X,Y)_{\leq t} = \pi_t}$ is known to both players at the $(t+1)$-step,
  and they therefore can think of their shared string of randomness as a list of samples from it. Assume without loss of generality that Alice speaks.
  Since her next message only depends on her input and the transcript so far, the distribution
  $\Pi(X_1,Y_1)_{t+1},\ldots,\Pi(X_k,Y_k)_{t+1}\big)|_{X=x, Y=y, \Pi(X,Y)_{\leq t} = \pi_t}$
  is identical to the distribution
  $\Pi(X_1,Y_1)_{t+1},\ldots,\Pi(X_k,Y_k)_{t+1}\big)|_{X=x, \Pi(X,Y)_{\leq t} = \pi_t}$, and
  in particular she knows it.
  Therefore, Alice can use the sampling procedure from Lemma~\ref{lem:rejection_samp} to pick an index $r^{\star}$ in their
  string of randomness that refers to the $r^{\star}$th sample in their shared randomness string, such that this sample
  is distributed according to $\Pi(X_1,Y_1)_{t+1},\ldots,\Pi(X_k,Y_k)_{t+1}\big)|_{X = x, Y=y, \Pi(X,Y)_{\leq t} = \pi_t}$.
  Alice can communicate
  $r^{\star}$ to Bob, and then they can continue. We note that the correctness of the sampling, as well as the expected communication
  cost of the $(t+1)$th step, trivially follow from Lemma~\ref{lem:rejection_samp}.

  Thus, $\Gamma_k$ is a zero-error protocol for $f^{k}$, and we next upper bound its expected communication complexity.
  Let $Z_1,\ldots,Z_k$ denote the number of rounds in the execution of $\Pi$ on $(X_1,Y_1),\ldots,(X_k,Y_k)$ respectively,
  and let $T = \max(Z_1,\ldots,Z_k)$.
  By our analysis of each round, we have that
  \begin{equation}\label{eq:10}
  \CC_{\mu^k}(\Gamma_k)\leq
  \sum\limits_{t=0}^{\infty}
  \Expect{x,y}{\Expect{\pi_{t}}{
  1_{T\geq t}\left(O(1) + D(\pi_{t},t,,x,y) + 2\log(D(\pi_{t},t,,x,y) + 1)\right)}}.
  \end{equation}
  Clearly, the first term contributes at most $O(\Expect{}{T})$, which by Claim~\ref{claim:max} is $o(k)$.
  Next, we upper bound the contribution from the second term, which will also allow us to bound the last term using Jensen's inequality.
  To upper bound the contribution of the second term in~\eqref{eq:10}, note that by the chain-rule for conditional KL-divergence we have that
  \begin{align*}
  &\Expect{x_i,y_i}{\sum\limits_{t=0}^{\infty}\Expect{\pi_{i,t}}{D(\pi_{t},t,x,y)}}\\
  &\qquad=\Expect{\vec{x},\vec{y}}{\DKL{\big(\Pi(X_1,Y_1),\ldots,\Pi(X_k,Y_k)\big)|_{X = x,Y=y}}{\Pi(X_1,Y_1),\ldots,\Pi(X_k,Y_k)}}.
  \end{align*}
  By Fact~\ref{fact:mutual_div_KL}, this is equal to $\MI[X,Y; \Pi(X_1,Y_1),\ldots,\Pi(X_k,Y_k)]$, which by
  independence between the $(X_i,Y_i)$ for different $i$'s, is equal to $k\MI[X_1,Y_1; \Pi(X_1,Y_1)] = k\MI^{\sf external}_{\mu}[\Pi]$.

  For the second term, we note that $\log(1+z)\leq \sqrt{z}$ and so
  \begin{align*}
  \sum\limits_{t=0}^{\infty}
  \Expect{x,y}{
  \Expect{\pi_{t}}{2\cdot 1_{T\geq t}\log(D(\pi_{t},t,x,y))}}
  &\leq 2\sum\limits_{t=0}^{\infty}\Expect{x,y}{\Expect{\pi_{t}}{1_{T\geq t}\sqrt{D(\pi_{t},t,x,y)}}}\\
  &\leq 2\sqrt{\sum\limits_{t=0}^{\infty}\Expect{x_i,y_i}{\Expect{\pi_{t}}{1_{T\geq t}}}}
  \sqrt{\sum\limits_{t=0}^{\infty}\Expect{x,y}{\Expect{\pi_{t}}{D(\pi_{t},t,x,y)}}},
  \end{align*}
  where the last inequality is by Cauchy-Schwarz. The first term in this product is $\sqrt{\E[T]} = o(\sqrt{k})$,
  and the second term in this product is $\sqrt{k \MI^{\sf external}_{\mu}[\Pi]}$, so overall this expression is $o(k)$.

  Plugging everything into~\eqref{eq:10}, we see that the expected communication complexity of $\Gamma_k$ is
  at most $k\MI^{\sf external}_{\mu}[\Pi] + o(k)$, as desired.
\end{proof}

\begin{claim}\label{claim:max}
  Suppose $Z$ is a non-negative, integer random variable with finite expectation, and let
  $Z_1,\ldots,Z_k$ be independent copies of $Z$. Then
  \[
  \lim_{k\rightarrow\infty}\frac{1}{k}\Expect{}{\max(Z_1,\ldots,Z_k)} = 0.
  \]
\end{claim}
\begin{proof}
  The assumption implies that $\E[Z] = \sum\limits_{r\geq 1}{\Prob{}{Z\geq r}} < \infty$, so
  $\lim_{R\rightarrow\infty} \sum\limits_{r\geq R}{\Prob{}{Z\geq r}} = 0$.
  Let $\eps>0$. Then there is $R$ (that depends on the distribution of $Z$)
  such that $\sum\limits_{r\geq R}{\Prob{}{Z\geq r}}\leq \eps$.

  Computing, we get that
  \[
  \frac{1}{k}\Expect{}{\max(Z_1,\ldots,Z_k)}
  =\frac{1}{k}\sum\limits_{r\geq 1}\Prob{}{\max(Z_1,\ldots,Z_k)\geq r}
  =\frac{1}{k}\sum\limits_{r\geq 1}1-\Prob{}{Z< r}^k.
  \]
  We split the sum into $r<R$ and $r\geq R$. The contribution from $r<R$ is clearly at most $R/k$,
  and by our earlier observation the contribution from $r\geq R$ is at most
  \[
  \frac{1}{k}\sum\limits_{r\geq R}1-\left(1-\Prob{}{Z \geq r}\right)^k
  \leq
  \frac{1}{k}\sum\limits_{r\geq R}k\Prob{}{Z\geq r}
  \leq \eps.
  \]
  It follows that
  $\limsup\frac{1}{k}\Expect{}{\max(Z_1,\ldots,Z_k)}\leq \eps$, and since this is true for all $\eps>0$ the claim is proved.
\end{proof}

\subsection{Non-deterministic external information complexity}
\label{sec:missing2}

In this section pe prove Theorem~\ref{thm:oneside}:

\smallskip
\noindent{\bf Theorem~\ref{thm:oneside}.}
 Let $\mu$ be a distribution with ${\sf supp}(\mu)\subseteq f^{-1}(1)$, then $$
\IC^{\sf external,1}_{\mu}[f,0]\le
\lim_{q\rightarrow \infty}	\CC_{\mu^q}^{1^q}[f^q,0]/q.
$$
\smallskip

Before proceeding, let us formally define the quantities in the theorem. For a function $f$ and an output $a$, we say that a distribution of messages $M=M(X,Y)$ along with acceptance functions $Acc_A:(X,M)\mapsto \{0,1\}$,
$Acc_B:(Y,M)\mapsto \{0,1\}$ is an $a$-proof for $f$, if the following properties hold for all $(x,y)$:
\begin{itemize}
	\item
	If $f(x,y)=a$, then for all $m$ with $\Pr[M=m|(x,y)]> 0$, $Acc_A(x,m)=1$ and $Acc_B(y,m)=1$;
	\item
		if $f(x,y)\neq a$, then for all $m$, $Acc_A(x,m)=0$ or $Acc_B(y,m)=0$.
\end{itemize}
Then the (average case) non-deterministic amortized communication complexity of a boolean $f(x,y)$ is defined as:
\begin{equation}
	\label{eq:nondet1}
		\CC_{\mu^q}^{1^q}[f^q,0]:= \inf_{M:~M\text{ is a $1^q$-proof for $f^q$}} \E_{(x,y)\sim\mu; m\sim M|_{(x,y)}} |m|.
\end{equation}
The non-deterministic external information complexity of $f$ is given by:
\begin{equation}
	\label{eq:nondet2}
	\IC^{\sf external,1}_{\mu}[f,0]:= \inf_{M:~M\text{ is a $1$-proof for $f$}} \MI_{(x,y)\sim\mu; m\sim M|_{(x,y)}}(XY;M).
\end{equation}
In other words, it's the smallest amount of information a proof that definitively convinces Alice and Bob that $f(x,y)=1$ can reveal about $(x,y)$
to an outside observer. With the definitions in place, we are ready to prove the theorem.

\begin{proof}[Proof of Theorem~\ref{thm:oneside}.]
	We will prove the inequality for any fixed $q$. Theorem~\ref{thm:oneside} then follows by taking $q\rightarrow\infty$. For a fixed $q$, let $M^q$ be a $1^q$-proof for $f^q$ such that $\E|M^q|$ realizes $	\CC_{\mu^q}^{1^q}[f^q,0]$.
	
	Fix an index $i\in[q]$. Let $M^q_i(x, y)$ be obtained as follows:
	(1) set $(x_i,y_i)=(x,y)$; (2) pick values $(x_j,y_j)\sim \mu$ for $j\neq i$; (3) sample $M^q$  conditioned on $(x_j,y_j)_{j=1}^q$.
	
	Note that since $M^q$ is a $1^q$-proof for $f^q$, $M^q_i$ is a $1$-proof for
	$f(x_i, y_i)$ for all $i$. Thus each $i$ gives rise to a valid $1$-proof
	for $f$. Importantly, to verify that $M^q_i$ is a valid proof that $f(x,y)=f(x_i,y_i)=1$ the players do not need to know the values of $(x_j, y_j)$ for $j\neq i$.
	
	To complete the proof we only need to show that there exists an $i$ such that
	\begin{equation}
	\label{eq:nondet3}
	\MI[M_i^q; XY] = \MI[M^q; X_iY_i] \le \E|M^q|/q.
	\end{equation}
	To see this, observe that
\begin{multline*}
 \E|M^q| \ge \HH(M^q) \ge \MI[M^q; X_1Y_1\ldots X_qY_q] = \sum_{i=1}^q
 \MI[M^q; X_iY_i | X_{<i}Y_{<i}]  \\
 =\sum_{i=1}^q \left(
 \MI[M^q  X_{<i}Y_{<i}; X_iY_i]-\MI[X_{<i}Y_{<i}; X_iY_i]\right) =
  \sum_{i=1}^q
 \MI[M^q  X_{<i}Y_{<i}; X_iY_i] \ge  \sum_{i=1}^q
 \MI[M^q  ; X_iY_i].
\end{multline*}
Therefore, there must exist and $i\in[q]$ such that \eqref{eq:nondet3} holds.
\end{proof}
\end{document}